\documentclass[a4article,fleqn,usenatbib]{mnras}

\usepackage{graphicx}

\usepackage{times}
\usepackage{aas_macros}
\usepackage{pdfpages}
\usepackage{grffile}
\usepackage{amsmath,amssymb}
\usepackage{amsfonts}
\usepackage{bm}
\usepackage{bbm}
\usepackage[ruled,linesnumbered]{algorithm2e}
\usepackage{bbold}
\usepackage{url}
 \usepackage{arydshln}

\usepackage{xcolor,colortbl}

\usepackage{tikz}
\usetikzlibrary{positioning}
\usepackage{multirow}
\usepackage{pifont}

\newtheorem{theorem}{Theorem}[section]

\newtheorem{proposition}[theorem]{Proposition}
\newtheorem{lemma}[theorem]{Lemma}
\newtheorem{remark}[theorem]{Remark}
\newtheorem{example}[theorem]{Example}
\newtheorem{corollary}[theorem]{Corollary}

\def\argmin{\mathop{\rm argmin}}
\def\argmax{\mathop{\rm argmax}}

\newcommand{\vect}[1]{\boldsymbol{#1}}

\newcommand{\tabL}{1.5mm}

\definecolor{Gray}{gray}{0.8}
\definecolor{Gray1}{gray}{0.9}
\newcolumntype{H}{>{\columncolor{Gray}}c}
\newcolumntype{I}{>{\columncolor{Gray1}}c}

\newcommand{\blue}[1]{\textcolor{blue}{#1}}

\SetKwRepeat{Do}{do}{while}

\title[]{Online radio interferometric imaging: assimilating and discarding visibilities on arrival}
\author[Cai, Pratley, and McEwen]{Xiaohao Cai$^{1}$\thanks{E-mail:~x.cai@ucl.ac.uk~(XC);~luke.pratley.15@ucl.ac.uk~(LP);\newline jason.mcewen@ucl.ac.uk (JDM)}, 
Luke Pratley$^{1}$\blue{\footnotemark[1]} and Jason D. McEwen$^{1}$\blue{\footnotemark[1]} \\
$^{1}$Mullard Space Science Laboratory,  University College London, Surrey RH5 6NT, United Kingdom  \\
} 

\begin{document}

\date{Accepted ---. Received ---; in original form ---}
\pagerange{\pageref{sec:intro}--\pageref{lastpage}}
\pubyear{2017}

\maketitle

\begin{abstract}
	The emerging generation of radio interferometric (RI) telescopes, such as the Square Kilometre Array (SKA), 
	will acquire massive volumes of data and transition radio astronomy to a big-data era.
	The ill-posed inverse problem of imaging the raw visibilities acquired by RI telescopes will become significantly 
	more computationally challenging, particularly in terms of data storage and computational cost. 
	Current RI imaging methods, such as CLEAN, its variants, and compressive sensing approaches ({\it i.e.} sparse regularisation),
	have yielded excellent reconstruction fidelity. However, scaling these methods to big-data remains difficult if not impossible 
	in some cases. All state-of-the-art methods in RI imaging lack the ability to process data streams as they are acquired during 
	the data observation stage. Such approaches are referred to as {\it online} processing methods. 
	We present an online sparse regularisation methodology for RI imaging. Image reconstruction is performed simultaneously with
	data acquisition, where observed visibilities are assimilated into the reconstructed image as they arrive and then discarded.  
	Since visibilities are processed online, good reconstructions are recovered much faster than standard (offline) methods
	which cannot start until the data acquisition stage completes. 
	Moreover, the online method provides additional computational savings and, most importantly, dramatically reduces 
	data storage requirements. 
	Theoretically, the reconstructed images are of the same fidelity as those recovered by the equivalent offline
	approach and, in practice, very similar reconstruction fidelity is achieved. 
	We anticipate online imaging techniques, as proposed here, will be critical in scaling RI imaging to the emerging 
	big-data era of radio astronomy.  
\end{abstract}


\section{Introduction}\label{sec:intro}
Since the 1930s, radio astronomy has transitioned from the first observations
to a data-rich era. Due to rapid technological developments, radio astronomy will transition from a data-rich era to the
so-called big-data era in coming years.
Radio interferometric (RI) telescopes allow us to explore the Universe by detecting radio waves emitted from a wide range of objects
in the sky. They observe the radio sky with high angular resolution and sensitivity, 
providing valuable information for astrophysics and cosmology \citep{RV46,ryl60,tho08}. 
Next-generation RI telescopes are being built to achieve science goals ranging from 
probing cosmic magnetic fields \citep{joh15} to the detection of the epoch of re-ionization \citep{koo15}, to name just a few.

Representative next-generation RI telescopes include:
the LOw Frequency ARray (LO-FAR\footnote{\url{http://www.lofar.org}}, \citealp{van13}),
the Extended Very Large Array (EVLA\footnote{\url{http://www.aoc.nrao.edu/evla}}),
the Australian Square Kilometre Array Pathfinder (ASKAP\footnote{\url{http://www.atnf.csiro.au/projects/askap}}, \citealp{hot14}), 
the Murchison Widefield Array (MWA\footnote{\url{http://www.mwatelescope.org/telescope}}, \citealp{tin13}), 
and the Square Kilometer Array (SKA\footnote{\url{http://www.skatelescope.org}}, \citealp{dew13}).
These new telescopes will acquire large volumes of data, and achieve significantly higher dynamic range and angular resolution 
than previous generations.
The SKA itself, for instance, will provide a considerable step in dynamic range -- 
six or seven orders of magnitude beyond prior telescopes -- and angular resolution, 
while producing massive volumes of data. For example, data rate estimates of SKA phase I are around five
terabits per second for both SKA1-low (a low frequency aperture array) and 
SKA1-mid (a mid frequency array of reflector dishes) \citep{BNB15}. Moreover, the data volume will be 
greater still in SKA phase II.

Briefly speaking, radio interferometers sample Fourier coefficients (visibilities) of the radio brightness distribution in the sky. 
Due to limited sampling in the Fourier plane, imaging an observation requires solving an ill-posed linear inverse 
problem \citep{tho08}, which is an important first step in many subsequent scientific analyses. 
The emerging era of big-data ushered in by the new generation of radio telescopes will, inevitably, 
bring further challenges in imaging and scientific analysis. The enormous data rates will create practical challenges in both storage 
and computation.  

Classical image reconstruction methods, such as CLEAN-based methods \citep{hog74,BC04,cor08,SFM11,off14,pra16} and the
maximum entropy method (MEM) \citep{A74,GD78,CE85}, have served the community well but do not exploit modern image reconstruction techniques.
They suffer from cumbersome parameter tuning and/or slow computation speed and require post processing due to their limited image models.
Furthermore, they struggle to confront the upcoming big-data era. Recently, compressive sensing (CS) techniques were ushered into RI imaging 
for image reconstruction
\citep{wia09a,wia09b,mce11,li11a,li11b,CMW12,car14,wol13,dab15,DWPMW17,DOAPMW17,gar15,OCRMTPW16,ODW17,PMdCOW16,KDTW17,KCTW17}.
CS-based methods exploit sparse regularisation techniques
and have shown promising results compared to traditional approaches such as CLEAN ({\it e.g.} \citealt{PMdCOW16,DOAPMW17,KDTW17}). 
Furthermore, several algorithms \citep{car14,OCRMTPW16,ODW17,KCTW17} 
have been developed to scale such approaches to big-data, as anticipated from the SKA, using, 
{\it e.g.}, distribution, parallelisation, dimensionality reduction, and/or stochastic strategies. 
In \cite{CPM17,CPM17B}, Bayesian inference techniques for sparsity-promoting priors were presented to quantify the 
uncertainties associated with reconstructed images, {\it e.g.} to estimate local credible intervals ({\it cf.} error bars) on
recovered pixels. In particular, in \cite{CPM17B}, {\it maximum-a-posteriori} (MAP) estimation
techniques were presented to scale uncertainty quantification to massive data sizes, {\it i.e.} to big-data. 

All of these reconstruction methods ({\it e.g.}, CLEAN, MEM and CS-based methods), however, store the 
entire set of observed visibilities for subsequent processing once data acquisition is completed 
({\it i.e.}, after the full observation is made).
In other words, they can be categorised as {\it offline} methods. In this article we develop 
{\it online} methods for RI imaging. {\it Online} methods process data piece-by-piece as they are acquired,
without having the entire data-set available from the start \citep{S11,H15}.
Processing is thus performed while the data are acquired.
In particular, online methods generally outperform offline methods in terms of computational efficiency,
and memory and storage costs.
In RI imaging, since the visibility acquisition process can take a reasonably long time (often $\sim$10 hours or longer)
and the observed visibilities require huge storage, particularly in the big-data era, online methods
can provide considerable advantages.

In this article we propose a special type of online imaging method for RI. 
The proposed online method is based on iterative convex optimisation algorithms \citep[\textit{e.g.}][]{CP10}, which 
have been applied to RI imaging problems already \citep[\textit{e.g.}][]{CPM17B}.
Compared with the standard (offline) algorithms 
which use visibilities from an entire observation at every iteration, our online method 
needs to deal only with one single visibility block acquired at the latest time slot, at each iteration.
Most importantly, a visibility block will be discarded once the online method assimilates it.
The storage space will be released at the same time and will then be used for next data block. 
Consequently, the storage requirements of our online algorithm are limited to the size of the visibility block,
rather than the size of the full set of visibilities acquired over an observation. Storage requirements of our online algorithm
are thus a very small fraction of the storage requirements of offline methods.
Once arriving at the last data block (the entire visibilities then have been observed) 
and after processing it, the image from the observation will have been reconstructed by our online method. 
At the moment the final visibilities are acquired, our online method is close to completing its reconstruction work, 
whereas standard methods are only able to start the reconstruction.
Moreover, we verify the convergence properties of our online method and show that the quality of the images 
recovered by our online method is essentially the same as the equivalent offline method. 
Although our proposed online method is customised for the application of RI imaging, 
the concept of the method itself is generic and therefore can be directly applied to 
many other applications, such as medical imaging.  
Furthermore, our online framework supports the uncertainty quantification methodology proposed in \cite{CPM17B}.
Both techniques can thus be combined to target image reconstruction and uncertainty analysis for the upcoming big-data era. 

The remainder of this article is organised as follows. In Section \ref{sec:ri} we review the RI imaging inverse problem,
MAP (maximum a posteriori) estimation, related state-of-the-art optimisation algorithms and some classical online
methods. Our general online optimisation algorithm is proposed in Section \ref{sec:alg-ol}, with a discussion of its convergence 
properties. Section \ref{sec:alg-ri} focuses on sparse image reconstruction for RI imaging using the proposed online method, including an analysis 
of visibility storage requirements and computational cost. Numerical results evaluating the performance of our
proposed method and the comparison with related methods are reported in Section \ref{sec:exp}. 
Finally, in Section \ref{sec:con}, we conclude with a brief description of our main contributions and 
future works.

\section{Radio interferometric imaging and related methods}\label{sec:ri}
In this section the inverse reconstruction problem of RI imaging is first reviewed. 
Then, we review MAP estimation techniques to address the RI imaging problem efficiently.
Finally, some representative online optimisation methods are reviewed. 

\subsection{Radio interferometry}
In the following, we briefly recall the background of the inverse problem of RI imaging 
(for further details see, {\it e.g.}, \citealt{CPM17} and references therein).
A radio interferometer consists of an array of radio antennae, where each pair of antennae forms a baseline.
When the field of view is narrow and the baselines are co-planar, the telescope measures visibilities, $\vect y$, 
by correlating the signals from pairs of antennas in a given baseline, 
with baseline vector $\vect u = (u, v)$ being defined as the separation of the antennae.

Let $\vect x$ represent the sky brightness distribution, described in coordinates $\vect l = (l, m)$, and $A(\vect l)$ represent the 
primary beam of the telescope. The RI measurement equation for obtaining $\vect y$ can be represented as  \citep{tho08}
\begin{equation} \label{eqn:ri}
	\vect y (\vect u) = \int A(\vect l) \vect x(\vect l) {\rm e}^{-2\pi i \vect u \cdot \vect l} {\rm d}^2 \vect l.
\end{equation}
Recovering the sky intensity signal $\vect x$ from the measured visibilities $\vect y$, acquired according to 
equation \eqref{eqn:ri}, forms a linear inverse problem \citep{RBVC09}.

In the discrete setting, let $\vect x \in \mathbb{R}^N$ represent the sampled intensity signal
(the sky brightness distribution) which, under a basis or dictionary ({\it e.g.}, a wavelet basis or an over-complete frame)
$\bm{\mathsf{\bm{\mathsf{\Psi}}}} \in \mathbb{C}^{N\times L}$, can be represented as
\begin{equation}\label{eqn:x}
	{\vect x} = \bm{\mathsf{\Psi}} {\vect a} = \sum_{i} \bm{\mathsf{\Psi}}_i a_i,
\end{equation}
where vector ${\vect a} = (a_1, \cdots, a_L)^\top$ represents the synthesis coefficients
of ${\vect x}$ under $\bm{\mathsf{\Psi}}$. In particular, ${\vect x}$ is said to be sparse if ${\vect a}$ contains only $K$ non-zero coefficients,
with $K\ll N$.  Similarly, $\vect x$ is called compressible if many coefficients of $\vect a$ are nearly zero, {\it i.e.}, 
its sorted coefficients $a_i$ satisfy a power law decay.
In practice, it is ubiquitous that natural images $\vect x$ are sparse or compressible.
Let $\vect y \in \mathbb{C}^M$ be the $M$ visibilities observed under a linear measurement operator
$\bm{\mathsf{\Phi}} \in \mathbb{C}^{M\times N}$ modelling the realistic acquisition of the sky brightness components. Then, we have
\begin{equation}\label{eqn:y}
	{\vect y}=\bm{\mathsf{\Phi}} {\vect x} + {\vect n},
\end{equation}
where ${\vect n} \in \mathbb{C}^{M}$ represents additive noise. Without loss of generality, we subsequently consider 
independent and identically distributed (i.i.d.) Gaussian noise. 
In practice, $\vect y$ is not sufficiently sampled, which results in an ill-posed inverse problem 
in the perspective of image reconstruction of $\vect x$ from \eqref{eqn:y}. 

For the upcoming big-data era, the number
of the data points in $\vect y$ could be much larger than the size of image $\vect x$, {\it i.e.}, $M \gg N$, providing additional challenges 
in processing speed, memory load and data storage requirements. These new challenges call for new efficient solvers, 
in addition to the distributed, parallel and stochastic computation approaches ({\it e.g.} \citealt{car14,OCRMTPW16}).

\subsection{Maximum a posteriori (MAP) estimation}
One very effective way to address the ill-posed inverse problem in \eqref{eqn:y} is by using
Bayes' theorem to infer the posterior distribution of the image $\vect x$ given data $\vect y$, by
\begin{equation}\label{eqn:baye}
	p(\vect x | \vect y) = \frac{p(\vect y | \vect x) p(\vect x)}{\int_{\mathbb{R}^N}p(\vect y | \vect s) p(\vect s) {\rm d} \vect s},
\end{equation}  
where $p(\vect y | \vect x)$ is the likelihood and $p(\vect x)$ is the prior.
The normalising constant in the denominator of \eqref{eqn:baye} is the marginal likelihood or Bayesian evidence, 
which need not be computed for parameter inference. The MAP estimator, a Bayesian point estimator, can be
obtained as a solution of problem \eqref{eqn:y} by 
\begin{equation}
	{\vect x}_{\rm map} = \argmax_{\vect x \in \mathbb{R}^N} p(\vect x | \vect y).
\end{equation}
In the following we derive a common objective functional used to produce ${\vect x}_{\rm map}$.
For the RI imaging problem of \eqref{eqn:y}, in the case of i.i.d. Gaussian noise, the likelihood function reads
\begin{equation}
	p(\vect y|\vect x) \propto {\rm exp}(-\|\vect y - {\cal A} \vect x \|_2^2 /2\sigma^2), 
\end{equation}
and let the prior distributions of $\vect x$ be
\begin{equation}
	p(\vect x) \propto {\rm exp}(-{\phi}({\cal B} \vect x)),
\end{equation}
where $\sigma$ represents the standard deviation of the noise, ${\cal A}$ and ${\cal B}$ are problem-related linear operators,
and $\phi$ encodes prior information of the image (acting as a regularising penalty). 
Refer to \cite{CPM17} for more detailed discussions regarding choices of the prior function $\phi$.
Let 
\begin{equation}\label{eqn:con-fun}
	{\cal F}_{\vect y} (\vect x) := \phi({\cal B} \vect x) + \|\vect y - {\cal A} \vect x \|_2^2 /2\sigma^2.
\end{equation}
Consider $\phi$ convex, from which it follows that ${\cal F}_{\vect y}$ is convex.
Then the inverse problem in \eqref{eqn:y} can be solved by the MAP estimator given by
\begin{equation}\label{eqn:x-map}
	{\vect x}_{\rm map} = \argmin_{\vect x \in \mathbb{R}^N} \left\{ \phi({\cal B} \vect x) + \|\vect y - {\cal A} \vect x \|_2^2 /2\sigma^2 \right\}.
\end{equation} 
Refer to \cite{CPM17} for further details about Bayesian inference in the context of RI imaging.

Let ${\cal B} = \bm{\mathsf{\Psi}}^\dagger$ (the adjoint of $\bm{\mathsf{\Psi}}$) and ${\cal A} = \bm{\mathsf{\Phi}}$
for the analysis setting and ${\cal B} = \bm{\mathsf{I}}$ (identity operator) and ${\cal A} = \bm{\mathsf{\Phi}}\bm{\mathsf{\Psi}}$ 
for the synthesis setting. After equipping $\phi$ with the $\ell_1$ norm, $\|\cdot\|_1$, to promote sparseness 
\citep{wia09a,wia09b,mce11,OCRMTPW16,PMdCOW16,CPM17B,CPM17},
the MAP estimation model in \eqref{eqn:x-map} reads
\begin{equation}\label{eqn:ir-un-af}
	{\vect x}_{\rm map} = \argmin_{\vect x} \Big\{\mu \|\bm{\mathsf{\Psi}}^\dagger {\vect x}\|_1 
	+ \|{\vect y}-\bm{\mathsf{\Phi}} {\vect x}\|_2^2/2\sigma^2 \Big\},
\end{equation}
or 
\begin{equation}\label{eqn:ir-un-sf}
	{\vect x}_{\rm map} = \bm{\mathsf{\Psi}}\times\argmin_{{\vect a}} \Big\{ \mu \|{\vect a}\|_1 
	+ \|{\vect y}-\bm{\mathsf{\Phi}}\bm{\mathsf{\Psi}} {\vect a}\|_2^2/2\sigma^2 \Big\},
\end{equation}
where $\mu$ is the regularisation parameter used to balance the tradeoff between sparsity and data fidelity. 
Models \eqref{eqn:ir-un-af} and \eqref{eqn:ir-un-sf} are generally coined as {\it analysis} and
{\it synthesis} unconstrained frameworks, respectively. Further discussions about the analysis and synthesis forms can be found in,
{\it e.g.}, \cite{MHL04,EMR07,CJP12,CPM17,CPM17B}. 

\subsection{Convex optimisation methods for MAP estimation} \label{sec:ri-co}
MAP estimation models like the analysis model \eqref{eqn:ir-un-af} and the synthesis model \eqref{eqn:ir-un-sf}
can be solved by convex optimisation methods. 

Consider a general problem represented as
\begin{equation} \label{eqn:fb-p}
	\argmin_{\vect x \in \mathbb{R}^{N}} \{ f(\vect x) + g_{\vect y}(\vect x)\},
\end{equation}
where $f: \mathbb{C}^{N} \rightarrow \mathbb{R}$ is proper, convex and lower semi-continuous, 
and $g_{\vect y}: \mathbb{C}^{N} \rightarrow \mathbb{R}$,
which is usually abbreviated as $g$ afterwards (when it is associated with all $\vect y$),
is convex, and differentiable with Lipschitz constant $\beta_{\rm Lip} \in (0, \infty)$, {\it i.e.},
\begin{equation}
	\|\nabla g(\hat{\vect z}) - \nabla g(\bar{\vect z})\| \le \beta_{\rm Lip} \|\hat{\vect z} - \bar{\vect z} \|, 
	\ \  \forall (\hat{\vect z}, \bar{\vect z}) \in \mathbb{C}^{N} \times \mathbb{C}^{N}.
\end{equation}
Define, $\forall \lambda \in \mathbb{R}^+$, the {\it proximity operator} of convex function $f$ at $\vect x \in \mathbb{R}^N$ as \citep{M65}
\begin{equation} \label{eqn:prox-ope}
	{\rm prox}_f^{\lambda} ({\vect x})  \equiv \argmin_{{\vect u}\in \mathbb{R}^N} \left \{ f({\vect u}) + \|{\vect u} - {\vect x}\|_2^2/2\lambda \right \},
\end{equation}
and represent ${\rm prox}_f^{1} ({\vect z})$ by ${\rm prox}_f({\vect z})$ for simplification.

The minimisation problem with form \eqref{eqn:fb-p} can be solved by many convex optimisation methods,
{\it e.g.}, the forward-backward splitting algorithm, the Douglas-Rachford splitting algorithm, the alternating direction method of multipliers (ADMM), 
or the simultaneous direction method of multipliers (SDMM) (see \citealt{CP10} and references therein).
In the following, we briefly recall the forward-backward algorithm, due to its simplicity, efficiency, and pertinence to the objective functionals
considered in this article, {\it i.e.},  the analysis model \eqref{eqn:ir-un-af} and synthesis model \eqref{eqn:ir-un-sf}.

In general, from the fixed point equation
\begin{equation} 
	{\vect x} = {\rm prox}_{\lambda f} \left({\vect x} - \lambda \nabla g({\vect x}) \right),  
\end{equation}
the iteration formula of the forward-backward algorithm can be written as \citep{CP10}
\begin{equation} \label{eqn:fb-i}
	{\vect x}^{(i+1)} = {\rm prox}_{ \lambda^{(i)} f} \left ({\vect x}^{(i)} -  \lambda^{(i)} \nabla g({\vect x}^{(i)}) \right),  
\end{equation}
where $ \lambda^{(i)}$ is the step size in a suitable bounded interval.
After sufficient iterations, sequence $\{{\vect x}^{(i)}\}$ generated by \eqref{eqn:fb-i}
 converges to a minimiser of problem \eqref{eqn:fb-p}.
There are several variants of \eqref{eqn:fb-i} that achieve more efficient performance;
refer to \cite{CW05}, \cite{BC11} and \cite{CPM17B} for more detailed discussions. 
 
\subsection{Online optimisation methods} 
Nowadays, various online problems have emerged in different research areas and application domains, such as
computer science and machine learning ({\it e.g.} online regression, pattern recognition); 
for more details refer to \cite{S11} and \cite{H15}, and references therein.
Even through these problems are different to our focus in this article, their intrinsic 
ideas, {\it i.e.}, their online manner, can help to design appropriate methods in new applications such as
RI imaging. We therefore, in the following, briefly recall some classic online optimisation methods. 

The general online optimisation protocol can be described as follows. 
In general, consider a set of possible actions ${\vect C}$, a set of incoming observations ${\vect D}$,
and a given loss function $l: {\vect C} \times {\vect D} \rightarrow \mathbb{R}^+$.
Here set ${\vect D}$ has a general meaning, not restricted to RI observations.
At each time slot $i = 1, 2, \ldots, n$: (i) choose an action ${\vect z}^{(i)} \in {\vect C}$ and simultaneously 
an observation ${\vect z}_*^{(i)} \in {\vect D}$; (ii) obtain the loss $l({\vect z}^{(i)}, {\vect z}_*^{(i)})$. 
Then, the objective of the online optimisation protocol is to select actions ${\vect z}^{(i)}$ to minimise the total loss
\begin{equation} \label{eqn:oom}
	\sum_i l({\vect z}^{(i)}, {\vect z}_*^{(i)}).
\end{equation}
The above problem \eqref{eqn:oom} can be solved by many methods, such as the online mirror descent,
the online Newton step algorithms, or online gradient methods (see \citealt{S11,H15} for more detail).
As an example, starting at ${\vect z}^{(1)} \in {\vect C}$, the online gradient method, for $i\ge1$, updates iteratively by
\begin{align}
	{\tilde{\vect z}}^{(i+1)} & = {\vect z}^{(i)} - \lambda \nabla l({\vect z}^{(i)}, {\vect z}_*^{(i)}), \label{eqn:ogm-1} \\
	{\vect z}^{(i+1)}         & = \argmin_{\vect z \in {\vect C}} \| {\tilde{\vect z}}^{(i+1)} - \vect z \|_2.               
\end{align} 

Note that in the iteration formula \eqref{eqn:fb-i} all observations acquired are needed at every iteration. 
On the contrary, the online iteration formula \eqref{eqn:ogm-1} shows that, at iteration $i$, only action ${\vect z}^{(i)}$
(observed at last iteration) and the latest incoming observation ${\vect z}_*^{(i)}$ (not all observations)
are used to derive a new action. For the online setting, all observations are only known at the final stage.
Seeing the visibility acquisition property in RI imaging in this manner, {\it i.e.}, the visibilities are observed one-by-one 
(or bunch-by-bunch), the online approach mentioned above can be adapted to tackle the RI imaging problem.
In the next section we propose our generic online method for minimisation problems such as \eqref{eqn:fb-p},
based on convex optimisation methods applicable for MAP estimation.

\section{Proposed online imaging method}\label{sec:alg-ol}
In practice, in RI the time of acquiring the measurements $\vect y$ can be long (often $\sim$10 hours or longer),
and the space needed to store the data can be extremely large, particularly in the big-data era.
Waiting to obtain and store all measurements is not efficient and imposes large costs that may be avoided. Furthermore,
popular RI imaging methods ({\it e.g.} CLEAN and CS-based methods) can require relatively long computation times
to recover images. There is an urgent need for online processing of the incoming data to recover images, which
is the main focus of this article. 
In this section we present our general online reconstruction method for
solving inverse imaging problems, {\it i.e.}, the analysis model \eqref{eqn:ir-un-af} and synthesis model \eqref{eqn:ir-un-sf},
with its convergence analysis. 
This new technique can significantly improve data processing efficiency in both computation time and data storage, 
and is essential to cope with the challenges of next-generation radio telescopes in the big-data era. 

The diagram in Figure \ref{fig:online_diag} shows the methodology of our proposed online method. 
As is shown, firstly, the algorithm checks whether the data observation stage has completed. If yes, no new data block will be 
observed and thus the online method stops. Otherwise, the algorithm: loads the new observed data block;
assimilates it; releases the data block; 
updates the intermediate reconstructed image (using the newly assimilated data); 
and then sets the current reconstructed image as the starting point for the next iteration. 
The above steps are repeated until the data observation stage completes and then the final reconstructed image is set as the output --
the reconstruction result of the online method.

\begin{figure}
  \begin{center}
    \begin{tabular}{c}
  	 \includegraphics[trim={{.55\linewidth} {1.03\linewidth} {.98\linewidth} {.87\linewidth}}, clip, width=0.8\linewidth, height = 0.9\linewidth]
		{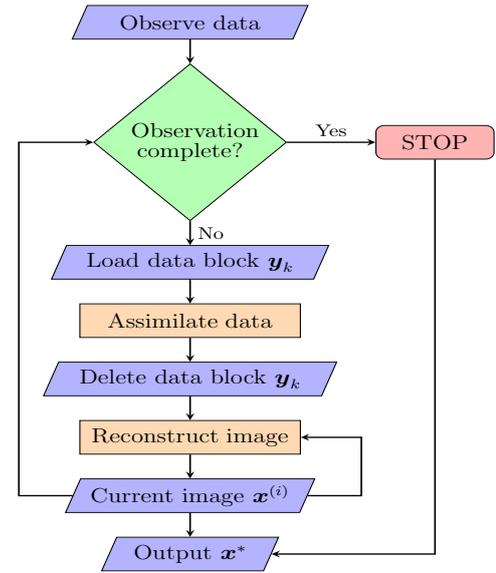} 
    \end{tabular}
  \end{center}
  \caption{Our proposed online imaging method ({\it e.g.} for RI imaging). 
  Firstly, the algorithm checks whether the data observation stage has completed. If yes, no new data block will be 
  observed and thus the online method stops. Otherwise, the algorithm: loads the new observed data block;
  assimilates it; releases the data block; 
  updates the intermediate reconstructed image (using the newly assimilated data); 
  and then sets the current reconstructed image as the starting point for the next iteration. 
  The above steps are repeated until the data observation stage completes and then the final reconstructed image is set as the output --
  the reconstruction result of the online method.   }
  \label{fig:online_diag}
\end{figure}

\subsection{Data blocks}
Without loss of generality, we split the measurements ${\vect y} \in \mathbb{C}^M$ into $B$ blocks and assume these blocks are received 
at different but consecutive time slots, {\it i.e.},
\begin{align} \label{eqn:split-y}
	{\vect y} & = \begin{bmatrix} 
	{\vect y}_1^{\top}, \cdots, {\vect y}_{k}^{\top}, \cdots,
	{\vect y}_{B}^{\top}
	\end{bmatrix}^{\top}, 
	\ \ {\vect y}_{k} \in \mathbb{C}^{M_k},
\end{align}
where block ${\vect y}_{k}$ is received earlier than ${\vect y}_{k+1}$ and $\sum_{k=1}^B M_k = M$.
Accordingly, we split the measurement operator $\bm{\mathsf{\Phi}}$ into $B$ blocks, {\it i.e.},
\begin{align}
	{ \bm{\mathsf{\Phi}}} & = \begin{bmatrix} 
	{ \bm{\mathsf{\Phi}}}_1^{\top},\cdots, { \bm{\mathsf{\Phi}}}_{k}^{\top}, \cdots, { \bm{\mathsf{\Phi}}}_{B}^{\top}
	\end{bmatrix}^{\top},
	\ \ { \bm{\mathsf{\Phi}}}_{k} \in \mathbb{C}^{M_k\times N}.
\end{align}
Note that, for each data block $k\in \{1, \cdots, B\}$, problem \eqref{eqn:y} can be rewritten as
\begin{equation}\label{eqn:y-p}
	{\vect y}_{k}=\bm{\mathsf{\Phi}}_{k} {\vect x} + {\vect n}_{k},
\end{equation}
where ${\vect n}_{k}$ represents additive noise associated with measurement block ${\vect y}_{k}$.

\subsection{Online algorithm}
For the $B$ blocks of ${\vect y}$ in \eqref{eqn:split-y}, which are assumed to be obtained at different times,
we suggest using the measurement blocks as they are received for early-stage reconstruction, 
rather than waiting and starting reconstruction once all measurement blocks are received. 
Then, once receiving the last measurement block, a complete reconstruction can be achieved immediately,
or at least much faster than an offline approach, since a close estimate of the underlying image is already available 
from early-stage reconstructions.
In the following, we present our online reconstruction strategy based on the standard 
forward-backward algorithm reviewed in the previous section.  

We assume $g$ (the second term in \eqref{eqn:fb-p}) is separable according to
the $B$ blocks of $\vect y$ in \eqref{eqn:split-y}, {\it i.e.}, 
\begin{equation} \label{eqn:split-f2}
	g = g_{{\vect y}_1} + \cdots + g_{{\vect y}_{k}} + \cdots + g_{{\vect y}_{B}},
\end{equation} 
where $g_{{\vect y}_{k}}$ represents function $g$ corresponding to 
data block $k$ in $\vect y$. For simplicity, we use $g_{k}$ to denote $g_{{\vect y}_{k}}$ in the following.
Note the fact that the data fidelity term in the analysis form \eqref{eqn:ir-un-af} or synthesis form \eqref{eqn:ir-un-sf} 
is indeed separable.
Then, using \eqref{eqn:split-f2}, the general minimisation problem \eqref{eqn:fb-p} can be rewritten as 
\begin{equation} \label{eqn:fb-p-ol}
	\argmin_{\vect x \in \mathbb{R}^{N}} {\cal F}_{\vect y} (\vect x) := \left \{ f(\vect x) + \sum_{k=1}^B g_{k}(\vect x) \right\}.
\end{equation} 

In the online setting, the $B$ blocks of $\vect y$ are being received one by one at different time slots, so the objective functional
\eqref{eqn:fb-p-ol} can only be exactly formed when all the $B$ blocks are acquired.
For the first $b$ blocks available at some time slot where $b \le B$, problem \eqref{eqn:fb-p-ol} can be written as 
\begin{equation} \label{eqn:fb-p-ol-part}
	\argmin_{\vect x \in \mathbb{R}^{N}} {\cal F}_{{\vect y}_1^b} (\vect x) :=  \left \{ f(\vect x) + \sum_{k=1}^{b} g_{k}(\vect x) \right \}.
\end{equation} 
Clearly, all the methods and techniques mentioned in Section \ref{sec:ri-co} that are applicable to solve problem \eqref{eqn:fb-p} 
can also be applied to solve problem \eqref{eqn:fb-p-ol-part}. 
For example, applying the forward-backward iterative formula \eqref{eqn:fb-i} to problem \eqref{eqn:fb-p-ol-part}, we obtain
\begin{equation} \label{eqn:fb-ol-i}
	{\vect x}^{(i+1)} = {\rm prox}_{ \lambda^{(i)} f} \left({\vect x}^{(i)} -  \lambda^{(i)} \nabla g^{b}_1 ({\vect x}^{(i)}) \right),
\end{equation}
where $g^{b}_{1} = g_{1} + \cdots + g_{b}$. 

With a given starting point, the iterative strategy of the proposed online algorithm is implemented as follows.  
Starting with $b = 1$ (corresponding to the first data block), execute formula \eqref{eqn:fb-ol-i} associated with objective functional
\eqref{eqn:fb-p-ol-part} for a few iterations (a single iteration can be applied in practice), 
and denote the current result by ${\vect x}^{(i)}$. Then, move to $b = 2$ (corresponding to the second data block)
at the appropriate time by executing formula \eqref{eqn:fb-ol-i} using ${\vect x}^{(i)}$ as the starting point. 
Continue this procedure until $b = B$ (corresponding to the final data block), and denote the final result by ${\vect x}^*$ --
the reconstruction result of the online method. 
In order to solve problem \eqref{eqn:fb-p-ol}, the online algorithm is tackling a series of subproblems \eqref{eqn:fb-p-ol-part} 
according to different $b$ ($1\le b \le B$), where the reconstruction result corresponding to the previous subproblem 
is used as a starting point to process the next one, until the final subproblem amounts to the full problem. 

We summarise our online algorithm in Algorithm \ref{alg:ofb}. 
While we present an overview of the general form of the algorithm, it is possible to separate the data assimilation and imaging 
stages and discard the data block as soon as the data are assimilated (as discussed in Section~\ref{sec:alg-ri}).
The stopping criterions used in Algorithm \ref{alg:ofb} are specified as type I and type II, and are defined as follows: 
(i) type I can be set using the maximum number of data blocks 
(if known in advance) or a feedback of whether no new data blocks are available;
(ii) type II can be set using a chosen maximum iteration number and, in practice, 
we set this value to 1, which requires the least computational cost.
Apart from the maximum iteration number, the relative error of the solutions at two consecutive iterations can also be adopted as 
a stopping criterion for type II. 

\begin{algorithm} 
	\caption{Online forward-backward algorithm} \label{alg:ofb}
	\textbf{Input:} ${\vect x}^{(0)} \in \mathbb{R}^N$, $\sigma$ and $\lambda^{(b)} \in (0, \infty)$\\
	\textbf{Output:} ${\vect x}^*$ \vspace{0.05in} \\ 
	$i = 0, b = 0$ \\
	\Do{Stopping criterion type I is not reached ({\it e.g.} maximum number of data blocks)}{
		$b=b+1$\\
		{\bf load} data ${\vect y}_{b}$  // Load data block \\ 
		\Do{Stopping criterion type II  is not reached [{\it e.g.} maximum iteration number (typically once) or relative error of the solutions]}{
			// Assimilate ${\vect y}_{b}$ and image \\
			${\vect x}^{(i+1)} = {\rm prox}_{ \lambda^{(b)} f} \Big({\vect x}^{(i)} - \lambda^{(b)} \nabla g^{s}_{b} ({\vect x}^{(i)}) \Big)$     \\
			$i = i + 1$
		}
		{\bf delete} ${\vect y}_{b}$   // Discard data block 
		}   \vspace{0.05in}
																								
	set ${\vect x}^* = {\vect x}^{(i)}$
\end{algorithm}

\begin{remark}
	The forward-backward algorithm presented in Algorithm \ref{alg:ofb} is just a specific example of our proposed 
	online methodology. The online iterative form, however, can be very general. Any 
	iterative methods which are applicable for minimising problem \eqref{eqn:fb-p-ol} are likely to be compatible 
	with the online strategy proposed here. In other words, these methods can be simply extended to 
	online versions in the same manner.
\end{remark}

\subsection{Convergence analysis}
Since the convergence properties of standard splitting algorithms have been verified ({\it e.g.} see \citealt{CP10} and references therein),
the convergence of our online algorithm is therefore self-evident if the stopping criterion type II in Algorithm \ref{alg:ofb} is 
proper, \textit{i.e.}, the maximum iteration number is assigned large enough. 
Nevertheless, it is still important to consider the converge performance when executing as few iterations as possible,
{\it e.g.} a single iteration in stopping criterion type II.

For simplicity, represent the iterative form \eqref{eqn:fb-ol-i} as
\begin{equation} 
	{\vect x}^{(i+1)} = {\cal K} ({\vect x}^{(i)}).
\end{equation}
In particular, symbol ${\cal K} $ here is a general operator which represents any type of iterative formula applicable.
Let ${\vect x}^*$ be a minimiser of problem \eqref{eqn:fb-p-ol}. We want to show that 
${\vect x}^{(i)}$ goes to ${\vect x}^*$ and ${\cal F}_{{\vect y}_1^b} ({\vect x}^{(i)})$ goes to ${\cal F}_{\vect y} ({\vect x}^*)$. 
Obviously, when $i\rightarrow \infty$ and $b \rightarrow B$, we have
\begin{equation} 
	{\vect x}^{(i)} \rightarrow {\vect x}^*, \quad
	{\cal F}_{{\vect y}_1^b} ({\vect x}^{(i)}) \rightarrow {\cal F}_{\vect y} ({\vect x}^*),
\end{equation}
due to the standard convergence results of splitting methods. 

For the online algorithm, the computed ${\vect x}^{(i+1)}$ are associated with just part of data blocks
before all of them are acquired. 
As ${\vect x}^{(i+1)}$ are finally expected to minimise problem \eqref{eqn:fb-p-ol}, there is an important relationship 
between ${\vect x}^{(i+1)}$ and the objective functional ${\cal F}_{\vect y}(\cdot)$, which considers all visibility blocks. 
To address this point, the following Theorem \ref{thm:cvg} concludes that sequence ${\cal F}_{\vect y}({\vect x}^{(i)})$ is
monotone decreasing to ${\cal F}_{\vect y}({\vect x}^*)$, under a mild assumption: 
let ${\vect x}^{(i+1)}$ be obtained with $b$ data blocks, $b \le B$, then assume
\begin{equation} \label{eqn:asm}
	\sum_{k=b+1}^{B} g_{k} ({\vect x}^{(i)}) \ge \sum_{k=b+1}^{B} g_{k} ({\vect x}^{(i+1)}).
\end{equation}
In words, inequality \eqref{eqn:asm} means that the intermediate reconstruction corresponding to a
later iteration fits the unknown data blocks better than the reconstruction corresponding to an earlier iteration.
This is reasonable, since the more data that is received and used, the better the intermediate reconstruction should fit the 
whole data.

\begin{theorem} \label{thm:cvg}
	Under assumption \eqref{eqn:asm} and let ${\vect x}^*$ be a minimiser of ${\cal F}_{\vect y}$ in \eqref{eqn:fb-p-ol}, 
	the sequence $\big\{{\cal F}_{\vect y}({\vect x}^{(i)}) \big\}_{i}$ produced by the online method 
	Algorithm \ref{alg:ofb} is monotone decreasing to ${\cal F}_{\vect y}({\vect x}^*)$.
\end{theorem}
\begin{proof}
	We just need to verify, $\forall i \in \mathbb{Z}$ (the set of all integers),
	\begin{equation}  \label{eqn:ieq}
		{\cal F}_{\vect y}({\vect x}^{(i)}) \ge {\cal F}_{\vect y}({\vect x}^{(i+1)}).
	\end{equation}
	Obviously, the above inequality \eqref{eqn:ieq} is correct if ${\vect x}^{(i+1)}$ is produced using all $B$ visibility blocks.
	Otherwise, if ${\vect x}^{(i+1)}$ is associated with $b$ visibility blocks, $b < B$, we have
	\begin{equation}  \label{eqn:ieq-i}
		{\cal F}_{{\vect y}_1^b}({\vect x}^{(i)}) \ge {\cal F}_{{\vect y}_1^b}({\vect x}^{(i+1)}).
	\end{equation}
	The above inequality \eqref{eqn:ieq-i} is obtained using the convergence property of splitting methods: the total 
	energy of the objective functional with respect to fixed visibilities (input) is monotone decreasing.
	Then, using \eqref{eqn:asm} and \eqref{eqn:ieq-i}, and the fact that
	\begin{equation} 
		{\cal F}_{\vect y}(\cdot) = {\cal F}_{{\vect y}_1^b}(\cdot) + \sum_{k=b+1}^{B} g_{k} (\cdot), 
	\end{equation}
	we have 
	\begin{align}
	{\cal F}_{\vect y}({\vect x}^{(i)})  & =  {\cal F}_{{\vect y}_1^b}({\vect x}^{(i)}) + \sum_{k=b+1}^{B} g_{k} ({\vect x}^{(i)}), \\
	& \ge  {\cal F}_{{\vect y}_1^b}({\vect x}^{(i+1)}) + \sum_{k=b+1}^{B} g_{k} ({\vect x}^{(i+1)}), \\ 
	 & = {\cal F}_{\vect y}({\vect x}^{(i+1)}).
	\end{align}
	Thus inequality \eqref{eqn:ieq} still holds if ${\vect x}^{(i+1)}$ is associated with $b$ visibility blocks, $b < B$. 
	This completes the proof.
\end{proof}

\section{Online RI Imaging}\label{sec:alg-ri}
In this section we present the explicit procedures of applying the general online method proposed 
in the previous section ({\it i.e.}, Algorithm \ref{alg:ofb}) to find MAP estimators for the RI imaging problem, 
using both the analysis form \eqref{eqn:ir-un-af} and synthesis form \eqref{eqn:ir-un-sf}.
Moreover, we also discuss the visibility storage requirements and computational costs
of our online method. We use the label \ $\bar{}$ \ for symbols related to the analysis model
and \ $\hat{}$ \ for symbols related to the synthesis model (following \citealt{CPM17,CPM17B}).

\subsection{MAP estimation by online convex optimisation}
\subsubsection{Analysis}
Set ${\bar f}({\vect x}) = \mu \|\bm{\mathsf{\Psi}}^\dagger {\vect x}\|_1$ and 
${\bar g}_{k}({\vect x}) = \|{\vect y}_{k} - \bm{\mathsf{\Phi}}_{k} {\vect x}\|_2^2/2\sigma^2$, $k = 1, \ldots, B$.
Then the reconstruction problem for the analysis form \eqref{eqn:ir-un-af}, $\forall b \in \{1, \cdots, B\}$, {\it i.e.},
\begin{equation} \label{eqn:fb-p-ol-part-ana}
	\min_{\vect x} \Big\{{\bar f}({\vect x}) + {\bar g}^{b}_1({\vect x})  \Big\},
\end{equation}
can be solved by the forward-backward  iteration formula given in \eqref{eqn:fb-ol-i}, {\it i.e.},
\begin{equation} \label{eqn:fb-i-ana}
	{\vect x}^{(i+1)} = {\rm prox}_{ \lambda^{(i)} {\bar f}} \left({\vect x}^{(i)} -  \lambda^{(i)} \nabla {\bar g}^{b}_1 ({\vect x}^{(i)}) \right).
\end{equation}
Assume  $\bm{\mathsf{\Psi}}^\dagger \bm{\mathsf{\Psi}} = \bm{\mathsf{ I}}$, where $\bm{\mathsf{ I}}$ is identity matrix. 
We have, $\forall \bar{\vect z} \in \mathbb{R}^N$,
\begin{equation} \label{eqn:prox-a}
	{\rm prox}_{\lambda {\bar f}} (\bar{\vect z}) = \bar{\vect z} + \bm{\mathsf{\Psi}} 
	\left ( {\rm soft}_{\lambda \mu}(\bm{\mathsf{\Psi}}^\dagger \bar{\vect z}) - \bm{\mathsf{\Psi}}^\dagger \bar{\vect z} \right ),
\end{equation}
and
\begin{equation} \label{eqn:grad-a}
	\nabla {\bar g}^{b}_1({\vect x}) = \sum_{k = 1}^{b} \bm{\mathsf{\Phi}}_{k}^\dagger 
	(\bm{\mathsf{\Phi}}_{k} {\vect x} - {\vect y}_{k})/\sigma^2,
\end{equation}
where
${\rm soft}_{\lambda}({\vect z}) = \big({\rm soft}_{\lambda}({z}_1), {\rm soft}_{\lambda}({z}_2), \cdots \big)$ 
is the pointwise soft-thresholding operator of vector ${\vect z}$ defined by
\begin{equation}
	{\rm soft}_{\lambda}({z}_k) =
	\begin{cases}
		{z}_k (|{z}_k|  -  \lambda)/|{z}_k| & {\rm if} \  |{z}_k| > \lambda, \\
		0                                   & {\rm otherwise}.               
	\end{cases} 
\end{equation}
Refer to \cite{CPM17} for the derivation of \eqref{eqn:prox-a}.
Substituting \eqref{eqn:prox-a} and \eqref{eqn:grad-a} into \eqref{eqn:fb-i-ana},
then the analysis problem \eqref{eqn:fb-p-ol-part-ana} can be solved iteratively by
\begin{align} 
	{\vect v}^{(i+1)} & = {\vect x}^{(i)} -  \lambda^{(i)} \sum_{k = 1}^{b} \bm{\mathsf{\Phi}}_{k}^\dagger                                                                                                               
	(\bm{\mathsf{\Phi}}_{k} {\vect x}^{(i)} - {\vect y}_{k})/\sigma^2, \label{eqn:ir-un-af-fb-i-1}  \\
	{\vect x}^{(i+1)} & = {\vect v}^{(i+1)} \! + \! \bm{\mathsf{\Psi}} \left ( {\rm soft}_{\lambda^{(i)} \mu}(\bm{\mathsf{\Psi}}^\dagger {\vect v}^{(i+1)}) \! -\! \bm{\mathsf{\Psi}}^\dagger {\vect v}^{(i+1)} \right ), 
	\label{eqn:ir-un-af-fb-i-2}
\end{align}
with initialisation set to, {\it e.g.} ${\vect x}^{(0)} = \bm{\mathsf{\Phi}}_{1}^{\dagger} {\vect y}_{1}$, {\it i.e.} the dirty image
according to the first data block. Note, importantly, that the term related to ${\vect y}_{k}$ in \eqref{eqn:ir-un-af-fb-i-1}, 
{\it i.e.}, $\bm{\mathsf{\Phi}}_{k}^{\dagger} {\vect y}_{k}$, can be computed once in advance.
\begin{remark} 
	In the analysis form \eqref{eqn:ir-un-af}, when choosing a $\bm{\mathsf{\Psi}}$ such that $\bm{\mathsf{\Psi}}^\dagger \bm{\mathsf{\Psi}} \neq \bm{\mathsf{ I}}$, 
	{\it i.e.} an over-complete basis $\bm{\mathsf{\Psi}}$, then ${\rm prox}_{\lambda {\bar f}} (\bar{\vect z})$ can be computed in an iterative manner: 
	\begin{align} 
		{\vect u}^{(t+\frac{1}{2})} & = \lambda_{\rm ite}^{(t)} ({\vect 1} \! - \! {\rm prox}_{\lambda \mu \|\cdot \|_1/\lambda_{\rm ite}^{(t)}}) 
		\left(\frac{{\vect u}^{(t-\frac{1}{2})}}{\lambda_{\rm ite}^{(t)}} \! + \! \bm{\mathsf{\Psi}}^{\dagger} {\vect u}^{(t)}  \right), \\
		{\vect u}^{(t+1)}           & = \bar{\vect z} - \bm{\mathsf{\Psi}} {\vect u}^{(t+\frac{1}{2})},                                           
	\end{align}
	where $\lambda_{\rm ite}^{(t)} \in (0, 2/\beta_{\rm Par})$ ($\beta_{\rm Par}$ is a constant satisfying 
	$\|\bm{\mathsf{\Psi}} \vect z \|^2 \le \beta_{\rm Par} \|\vect z \|^2, \forall \vect z \in \mathbb{R}^L$)  
	is a predefined step size and ${\vect u}^{(t)} \rightarrow {\rm prox}_{\lambda {\bar f}} (\bar{\vect z})$;
	refer to \citet{FS09} and \citet{JHF11} and references therein for details.
	Here, we repeat the remark given in \cite{CPM17,CPM17B} for ease of reference.
\end{remark}

\subsubsection{Synthesis}
Set $\hat{f}(\vect a) = \mu \|{\vect a}\|_1$ and 
$\hat{g}_{k}(\vect a) = \|{\vect y}_{k}-\bm{\mathsf{\Phi}}_{k} \bm{\mathsf{\Psi}} {\vect a}\|_2^2/2\sigma^2$.
Then, similar to \eqref{eqn:fb-p-ol-part-ana}, the reconstruction problem for the synthesis form \eqref{eqn:ir-un-sf}, $\forall b \in \{1, \cdots, B\}$, {\it i.e.},
\begin{equation} \label{eqn:fb-p-ol-part-syn}
	\min_{\vect x} \Big\{{\hat f}({\vect a}) + {\hat g}^{b}_1({\vect a})  \Big\},
\end{equation}
can be solved by the forward-backward  iteration formula given in \eqref{eqn:fb-ol-i}, {\it i.e.},
\begin{equation} \label{eqn:fb-i-syn}
	{\vect a}^{(i+1)} = {\rm prox}_{ \lambda^{(i)} {\hat f}} \left({\vect a}^{(i)} -  \lambda^{(i)} \nabla {\hat g}^{b}_1 ({\vect a}^{(i)}) \right). 
\end{equation}
We have, $\forall \hat{\vect z} = (\hat{z}_1, \cdots, \hat{z}_L) \in \mathbb{R}^L$,
\begin{equation} \label{eqn:prox}
	\begin{split}
		{\rm prox}_{\lambda  {\hat f}} (\hat{\vect z})  
		&=  {\rm prox}_{\lambda \mu \|\cdot \|_1} (\hat{\vect z})   \\
		&=  \argmin_{{\vect u} \in \mathbb{R}^L}  \lambda \mu \|{\vect u}\|_1 + \|{\vect u} - \hat{\vect z}\|^2/2   \\
		&= {\rm soft}_{\lambda \mu}(\hat{\vect z})
	\end{split}
\end{equation}
and 
\begin{equation} \label{eqn:grad}
	\nabla {\hat g}^{b}_1 ({\vect a}) =   \sum_{k = 1}^{b} \bm{\mathsf{\Psi}}^\dagger \bm{\mathsf{\Phi}}_{k}^\dagger 
	(\bm{\mathsf{\Phi}}_{k}\bm{\mathsf{\Psi}} {\vect a} - {\vect y}_{k})/\sigma^2.
\end{equation}
Substituting \eqref{eqn:prox} and \eqref{eqn:grad} into \eqref{eqn:fb-i-syn}, we get the iterative formula solving the synthesis
problem \eqref{eqn:fb-p-ol-part-syn}, {\it i.e.},
\begin{equation} \label{eqn:ir-un-sf-fb-i}
	\begin{split}
		{\vect a}^{(i+1)} & = {\rm soft}_{\lambda^{(i)} \mu} \Big ({\vect a}^{(i)}    \\
		& \qquad - \lambda^{(i)}  \sum_{k = 1}^{b} \bm{\mathsf{\Psi}}^\dagger  
		\bm{\mathsf{\Phi}}_{k}^\dagger (\bm{\mathsf{\Phi}}_{k}\bm{\mathsf{\Psi}} {\vect a}^{(i)} - {\vect y}_{k})/\sigma^2 \Big).  
	\end{split}	
\end{equation}
Like \eqref{eqn:grad-a}, the term related to ${\vect y}_{k}$ in \eqref{eqn:ir-un-sf-fb-i},
{\it i.e.}, $\bm{\mathsf{\Psi}}^\dagger \bm{\mathsf{\Phi}}_{k}^\dagger {\vect y}_{k}$,
can be computed once in advance.

\begin{remark}
	Note that, in \eqref{eqn:ir-un-af-fb-i-1} and \eqref{eqn:ir-un-sf-fb-i}, the operators $\bm{\mathsf{\Phi}}_{k}^\dagger\bm{\mathsf{\Phi}}_{k}$ and 
	$\bm{\mathsf{\Psi}}^\dagger \bm{\mathsf{\Phi}}_{k}^\dagger \bm{\mathsf{\Phi}}_{k}\bm{\mathsf{\Psi}}$ can be precomputed for later invoking.
	Most importantly, as already noted, the terms $\bm{\mathsf{\Phi}}_{k}^\dagger {\vect y}_{k}$ (the so-called dirty map according to the $k$-th data block) 
	and $\bm{\mathsf{\Psi}}^\dagger \bm{\mathsf{\Phi}}_{k}^\dagger {\vect y}_{k}$ respectively 
	in \eqref{eqn:ir-un-af-fb-i-1} and \eqref{eqn:ir-un-sf-fb-i}, for $k= 1, \ldots, b$, can also be computed just once for subsequent use.
\end{remark}

We summarise the online forward-backward splitting algorithm corresponding to the analysis and synthesis reconstruction forms 
in Algorithms \ref{alg:ofb-a} and \ref{alg:ofb-s}, in which the stopping criteria mentioned in Algorithm \ref{alg:ofb} are specified explicitly. 
In particular, when processing individual visibility blocks, for simplicity and efficiency,
we execute one iteration ({\it cf.} stopping criterion type II in Algorithm  \ref{alg:ofb}) 
for each block.
Furthermore, when the number of visibility blocks $B$ is small (less than the iteration number necessary for a standard
forward-backward algorithm to converge), a few more optional iterations can then be applied after the algorithms process 
the last visibility block. 
However, in practice $B$ is generally very large in order to reduce storage costs, especially with the trend towards big-data,
in which case these optional iterations are not necessary. 

\begin{algorithm} 
	\caption{Online forward-backward  algorithm for the analysis model \eqref{eqn:ir-un-af}} \label{alg:ofb-a}
	\textbf{Input:} ${\vect x}^{(0)} \in \mathbb{R}^N$, $\sigma$ and $\lambda^{(b)} \in (0, \infty)$\\
	\textbf{Output:} ${\vect x}^*$ \vspace{0.05in} \\ 
	$i = 0, b = 0, {\vect v}_{\rm t1} = 0$  \vspace{0.05in} \\
	// Online update \\
	\Do{New visibility block}{
		$b=b+1$\\
		{\bf load} visibility ${\vect y}_{b}$ \\ 
		${\vect v}_{\rm t1} = {\vect v}_{\rm t1} + \lambda^{(b)} \bm{\mathsf{\Phi}}_{b}^\dagger {\vect y}_{b}/\sigma^2$ \\
		{\bf delete} visibility ${\vect y}_{b}$ \\
		${\vect v}_{\rm t2} = \lambda^{(b)} \sum_{k = 1}^{b} \bm{\mathsf{\Phi}}_{k}^\dagger \bm{\mathsf{\Phi}}_{k} {\vect x}^{(i)}/\sigma^2$ \\
		update ${\vect v}^{(i+1)} = {\vect x}^{(i)} + {\vect v}_{\rm t1} - {\vect v}_{\rm t2}   $ \\
		compute $\vect u = \bm{\mathsf{\Psi}}^\dagger {\vect v}^{(i+1)}$ \\
		update ${\vect x}^{(i+1)} = {\vect v}^{(i+1)} + \bm{\mathsf{\Psi}} \left ( {\rm soft}_{\lambda^{(i)} \mu}(\vect u) - \vect u \right )$  \\
		$i=i+1$ 
		}      \vspace{0.05in}
																												
	// Update with all assimilated visibilities (optional) \\
	\While{Maximum iteration number is not reached (cf. stopping criterion type II in Algorithm \ref{alg:ofb})} {
		${\vect v}_{\rm t2} = \lambda^{(b)} \sum_{k = 1}^{b} \bm{\mathsf{\Phi}}_{k}^\dagger \bm{\mathsf{\Phi}}_{k} {\vect x}^{(i)}/\sigma^2$ \\
		update ${\vect v}^{(i+1)} = {\vect x}^{(i)} + {\vect v}_{\rm t1} - {\vect v}_{\rm t2}   $ \\
		compute $\vect u = \bm{\mathsf{\Psi}}^\dagger {\vect v}^{(i+1)}$ \\
		update ${\vect x}^{(i+1)} = {\vect v}^{(i+1)} + \bm{\mathsf{\Psi}} \left ( {\rm soft}_{\lambda^{(i)} \mu}(\vect u) - \vect u \right )$  \\
		$i=i+1$
	}
	\vspace{0.05in}
																														
	set ${\vect x}^* = {\vect x}^{(i)}$ 
\end{algorithm}
\begin{algorithm} 
	\caption{Online forward-backward  algorithm for the synthesis model \eqref{eqn:ir-un-sf}} \label{alg:ofb-s} 
	\textbf{Input:} ${\vect a}^{(0)} \in \mathbb{R}^L$, $\sigma$ and $\lambda^{(i)} \in (0, \infty)$\\
	\textbf{Output:} ${\vect a}^*$ \vspace{0.05in} \\ 
	$i = 0, b = 0, {\vect v}_{\rm temp} = 0$  \vspace{0.05in} \\
	// Online update \\
	\Do{New visibility block}{
		$b=b+1$ \\
		{\bf load} visibility ${\vect y}_{b}$ \\ 
		${\vect v}_{\rm t1} = {\vect v}_{\rm t1} + \lambda^{(b)}  \bm{\mathsf{\Psi}}^\dagger \bm{\mathsf{\Phi}}_{b}^\dagger {\vect y}_{b}/\sigma^2$ \\
		{\bf delete} visibility ${\vect y}_{b}$ \\
		${\vect v}_{\rm t2} = \lambda^{(b)} \sum_{k = 1}^{b} \bm{\mathsf{\Psi}}^\dagger  
		\bm{\mathsf{\Phi}}_{k}^\dagger \bm{\mathsf{\Phi}}_{k}\bm{\mathsf{\Psi}} {\vect a}^{(i)}/\sigma^2$ \\
		update ${\vect a}^{(i+1)} = {\rm soft}_{\lambda^{(b)} \mu} ({\vect a}^{(i)} + {\vect v}_{\rm t1} - {\vect v}_{\rm t2} ) $  \\
		$i=i+1$
		}   \vspace{0.05in}
																																		
	// Update with all assimilated visibilities (optional) \\
	\While{Maximum iteration number is not reached (cf. stopping criterion type II in Algorithm \ref{alg:ofb})}{
		${\vect v}_{\rm t2} = \lambda^{(b)} \sum_{k = 1}^{b} \bm{\mathsf{\Psi}}^\dagger  
		\bm{\mathsf{\Phi}}_{k}^\dagger \bm{\mathsf{\Phi}}_{k}\bm{\mathsf{\Psi}} {\vect a}^{(i)}/\sigma^2$ \\
		update ${\vect a}^{(i+1)} = {\rm soft}_{\lambda^{(b)} \mu} ({\vect a}^{(i)} + {\vect v}_{\rm t1} - {\vect v}_{\rm t2} ) $  \\
		$i=i+1$
		} \vspace{0.05in}
																																							set ${\vect a}^* = {\vect a}^{(i)}$
																																		
\end{algorithm}

\addtolength{\tabcolsep}{-\tabL}
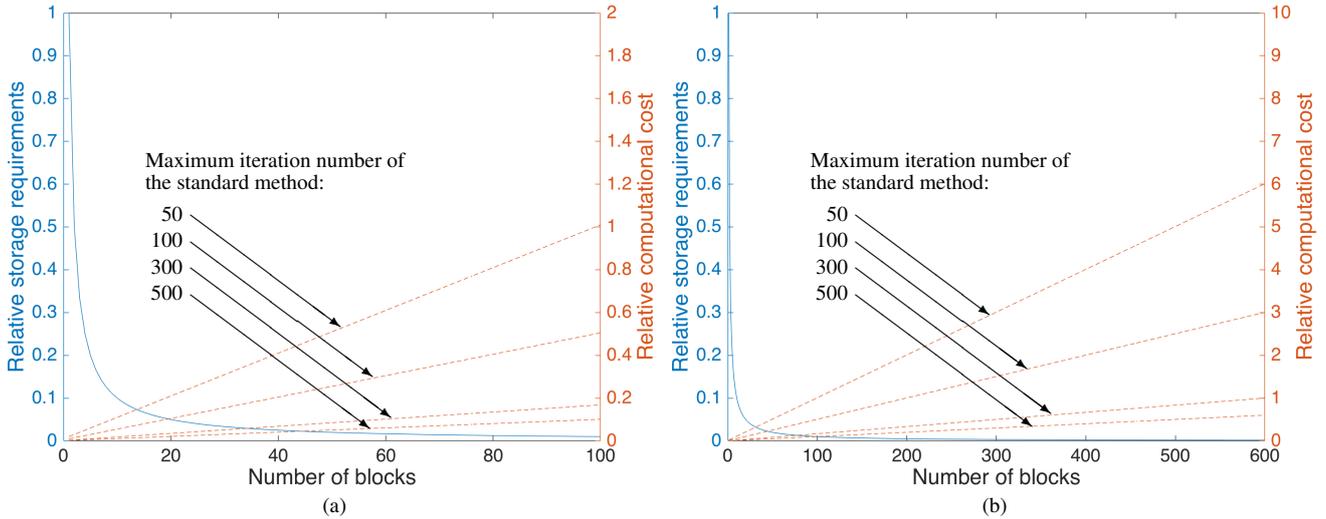
\begin{figure*}
	\centering
	\begin{tabular}{cc}
		\includegraphics[trim={{.055\linewidth} {.02\linewidth} {.02\linewidth} {.05\linewidth}}, clip, width=0.48\linewidth, height = 0.36\linewidth]
		{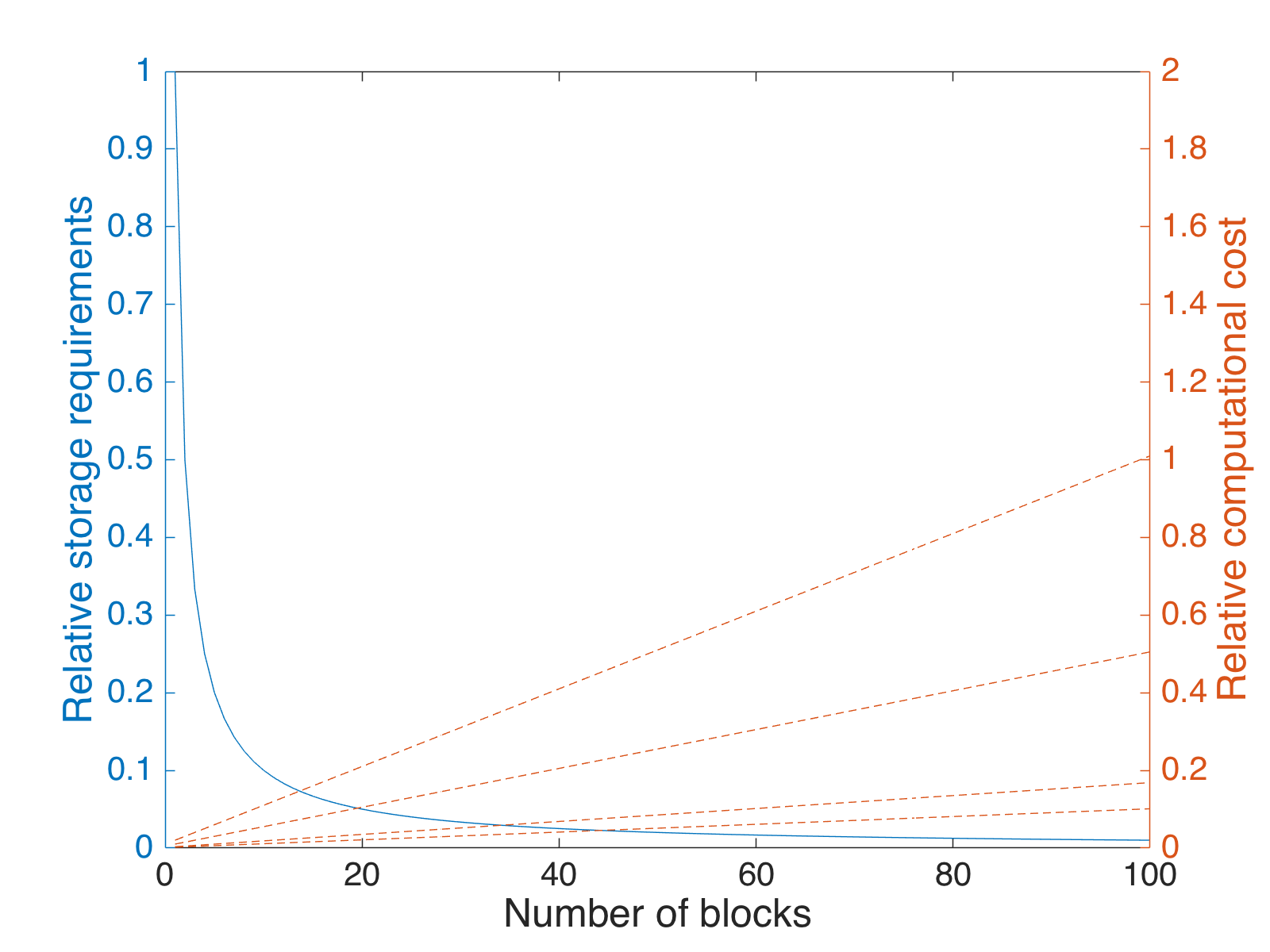} 
		\put(-190,120){Maximum iteration number of} 
		\put(-190,112){the standard method:} 
		\put(-186,100){ 50}  \put(-175,102){ \vector(4,-3){56}} 
		\put(-190,90){ 100}  \put(-175,92){ \vector(4,-3){68}} 
		\put(-190,80){ 300} \put(-175,82){ \vector(4,-3){75}} 
		\put(-190,70){ 500} \put(-175,72){ \vector(4,-3){67}}
		\hspace{0.05in}
		    &     
		\includegraphics[trim={{.055\linewidth} {.02\linewidth} {.02\linewidth} {.05\linewidth}}, clip, width=0.48\linewidth, height = 0.36\linewidth]
		{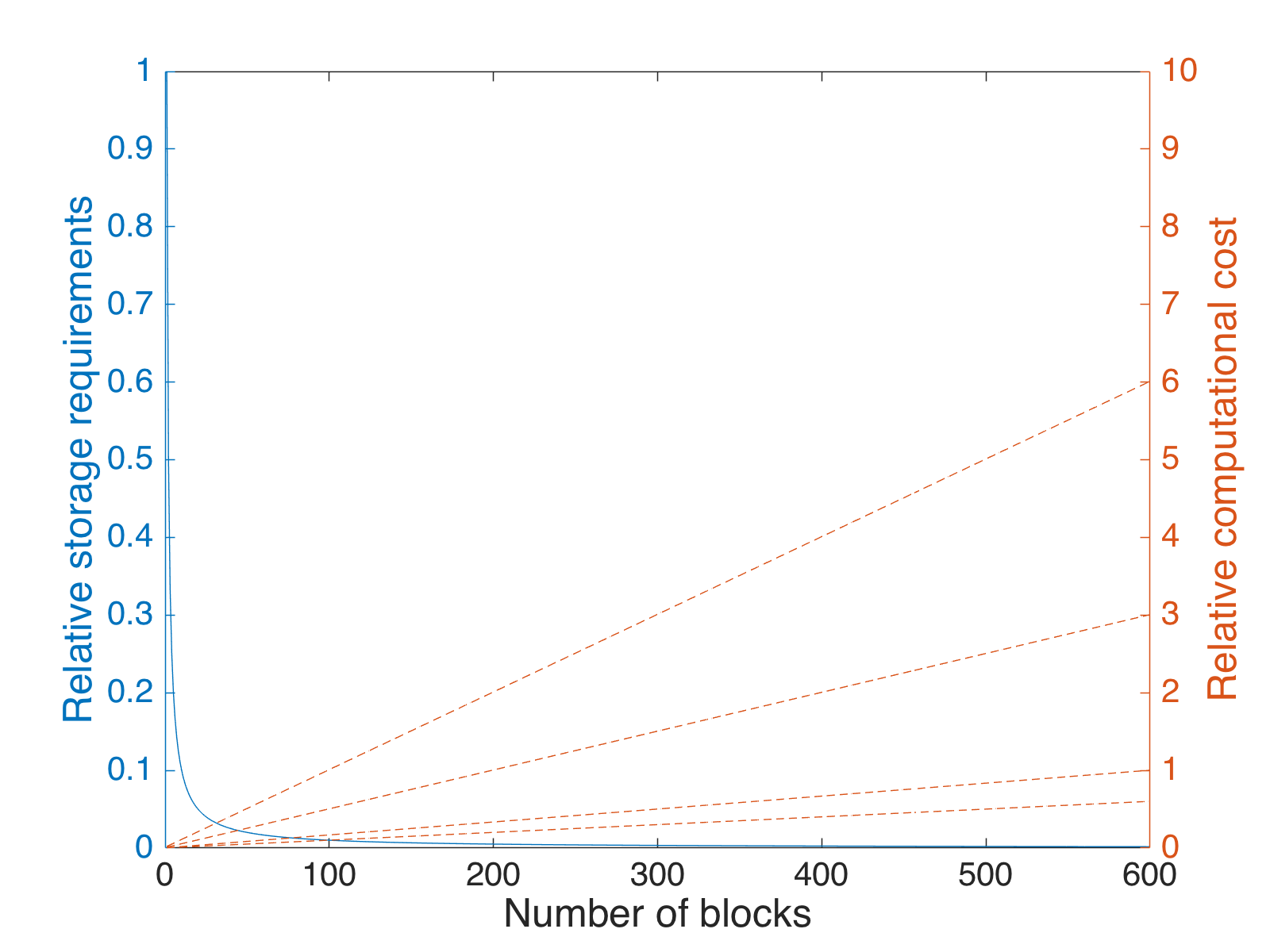} 
		\put(-190,120){Maximum iteration number of} 
		\put(-190,112){the standard method:} 
		\put(-186,100){ 50}  \put(-175,102){ \vector(4,-3){50}} 
		\put(-190,90){ 100}  \put(-175,92){ \vector(4,-3){64}} 
		\put(-190,80){ 300} \put(-175,82){ \vector(4,-3){73}} 
		\put(-190,70){ 500} \put(-175,72){ \vector(4,-3){66}} \\
		(a) & (b) 
	\end{tabular} 
	\caption{ Comparison between the standard algorithm and the online algorithm (this work) in terms of visibility storage requirements
		and computational cost. In the plots, the left vertical-axis represents the ratio of visibility storage requirements 
		(described in the text as $\eta_{\rm s} = 1/B$ when all blocks have the same size) between the online algorithm 
		with different number of visibility blocks and the standard algorithm (blue solid curve); the right vertical-axis represents the approximate 
		ratio of computational cost between the online algorithm and the standard algorithm with 
		different maximum iteration numbers, {\it i.e.}, $i_{\rm max} = 50, 100, 300$ and $500$ (brown dashed lines). 
		In particular, panels (a) and (b) correspond to a maximum number of visibility blocks set to 100 and 600, respectively.
		These plots show that as the number of visibility blocks increases the online method needs significantly less storage   
		than the offline method. The computational cost can also be reduced by approximately a half using the online method when both
		methods execute similar number of iterations. 
	}
	\label{fig-ite-comp-memory}
\end{figure*}
\addtolength{\tabcolsep}{\tabL}

\subsection{Storage requirements}
We discuss visibility storage requirements of the proposed online method in RI imaging.
In the forthcoming big-data era storing the visibilities $\vect y$ will be challenging.
Standard offline methods, which need all visibilities to be acquired and stored in advance, have extremely large storage requirements.
The proposed online method for RI imaging, as illustrated in Algorithms \ref{alg:ofb-a} and \ref{alg:ofb-s},
can dramatically reduce storage requirements. 

The online method only needs to deal with a single block of visibilities ({\it i.e.}, a subset of the visibilities) at one time. 
The size of each block can be controlled as required: when a large storage volume is available, a large visibility block can be considered; 
otherwise, any arbitrarily small block can be considered, to the extreme case of just a single visibility in each block (see 
lines 7 and 8 respectively in Algorithms \ref{alg:ofb-a} and \ref{alg:ofb-s} about the visibility block loading and assimilation).
Note that after loading and assimilating a block by the online method, the storage used to store that block will be released 
for storing another block (see line 9 in Algorithms \ref{alg:ofb-a} and \ref{alg:ofb-s} about the visibility block storage releasing). 
The ratio of visibility storage required for the online method relative to the offline method, which must store all $M$ visibilities,
is therefore
\begin{equation} \label{eqn:storage}
\eta_{\rm s} = \frac{\max_k\{M_k\}}{M}.
\end{equation} 
When all blocks are the same size, the storage requirement is \mbox{$\eta_{\rm s} = 1/B$} of the total visibilities, 
which means less than 1 percent of visibilities need to be stored when $B>100$.  
Figure \ref{fig-ite-comp-memory} (the blue solid curve) shows the ratio of visibility storage requirements 
between the online algorithm and the standard algorithm for different number of visibility blocks. 

Another important advantage of the online method in terms of storage requirement is that, due to its independence 
with respect to the number of visibility blocks, it has the ability of tackling RI imaging
problems encountered with an arbitrarily large amount of visibilities -- just divide the entire visibilities into 
individual visibilities blocks and then conquer them one-by-one online.  

Finally, since the standard offline methods can only deal with a complete set of visibilities, when new 
visibilities are available it is not possible for standard methods to use the new input to improve their reconstruction quality
in a principled manner (unless the computation is restarted). 
The online method, on the contrary, is able to immediately process any new observed visibilities 
-- just treat the new input as a normal visibility block and assimilate it to update the reconstruction.

It should be noted that storage during the image reconstruction process is not only burdened by the measurements, but also by storing 
the baseline coordinates and weights, which are comparable. Furthermore, the interpolation kernel for performing a  two dimensional 
non-uniform fast Fourier Transform can take up to 16 or more times the amount of storage from the 
measurements alone \citep{fes03,off14,PMdCOW16}. However, methods exist to reduce this storage cost. For example, 
it is possible for the interpolation kernel to be calculated on-the-fly, 
or to prune the interpolation kernel. Furthermore, alternative efficient methods can be developed to reduce these
storage costs, {\it e.g.} by precomputing $\bm{\mathsf{\Phi}}^\dagger\bm{\mathsf{\Phi}}$.

\subsection{Computational cost}
We now compare the computational cost between the online method and the standard 
method. Comparing to the standard method, in addition to dramatically reducing storage costs,
the online method can provide considerable computational savings when the number of visibility blocks considered 
is not much larger than the number of iterations necessary for the standard method.

For both the online and standard methods, at each iteration, the most computationally demanding part 
is to apply the measurement operator $\{\bm{\mathsf{\Phi}}_{k}^\dagger\bm{\mathsf{\Phi}}_{k}\}_{k=1}^b$ on an image
(refer to line 10 in Algorithms \ref{alg:ofb-a} and \ref{alg:ofb-s}), 
in that the rest of the computations are highly dominated by this step. 
In particular, for this step the standard method needs 
all the $B$ blocks, {\it i.e.}, $\{\bm{\mathsf{\Phi}}_{k}^\dagger\bm{\mathsf{\Phi}}_{k}\}_{k=1}^B$, to be involved in 
the computation for each iteration, whereas only the first $b$ blocks, $b < B$, are used at the $b$-th iteration 
in the online method.  
In other words, for the online method at iteration $b$, just $b/B$ portion of $\bm{\mathsf{\Phi}}^\dagger$ and $\bm{\mathsf{\Phi}}$ 
is involved in the computation (again, as an example, refer to line 10 in Algorithm \ref{alg:ofb-a}), whereas the 
standard method needs the whole operators. When we consider our online method with $B$ iterations (to ensure convergence $B\gg 1$), 
and where $i_{\rm max}$ denotes the maximum iteration number of the standard method, 
the ratio of the computational cost between the online method and the standard method can be approximated by
\begin{equation} \label{eqn:com-ratio}
	\eta_{\rm c} = \frac{\sum_{b=1}^B b/B} {i_{\rm max}} =  \frac{(B+1)/2}{i_{\rm max}},
\end{equation}
where the numerator and denominator represent the computational costs of the online method and the standard
method, respectively. Figure \ref{fig-ite-comp-memory} (the brown dashed lines) shows the approximate ratio of the computational cost between 
the proposed online method (with different number of visibility blocks, $B$) and the standard method 
(with different maximum iteration numbers, $i_{\rm max}$).

From \eqref{eqn:com-ratio}, we conclude that, when $i_{\rm max}$ is large enough ({\it i.e.}, ${1}/{i_{\rm max}} \approx 0$) 
and both methods execute similar iterations ({\it i.e.}, $B \approx {i_{\rm max}}$), the online method then requires only 
a half of the computations of the standard method, approximately, {\it i.e.},
\begin{equation} 
	\eta_{\rm c} \approx \frac{1}{2}.
\end{equation}
Furthermore, ratio \eqref{eqn:com-ratio} also reveals the following twofold fact: one is that if $1 \ll B \ll {i_{\rm max}}$, 
the online method will be much more economical; the other is that if $B \gg {i_{\rm max}}$, the online method will be 
computationally expensive, but then its visibility storage requirement is extremely low -- 
just approximate $\eta_{\rm s} = 1/B$ of the storage space needed by the standard method (see \eqref{eqn:storage}).
On the whole, the online method provides the choice of 
using a large or small number of visibility blocks $B$. This choice depends on the priority of the application, 
in saving storage space or reducing computational cost. 
Again, see Figure \ref{fig-ite-comp-memory} for the pictorial explanation of the comparison between the online 
and standard methods in terms of visibility storage requirements and computational cost.

Finally, it is worth emphasising that the online method actually achieves a reconstruction once 
no more visibility blocks are available. 
In other words, no matter what the computational costs, the online method executes them
before the data acquisition stage finishes. On the contrary, all of the computational costs of the standard method 
have to be carried out after the data acquisition stage. In this sense the online method always wins at the starting point
of the offline method.

\section{Experimental results}\label{sec:exp}
In this section we investigate the performance of the proposed online method 
using representative RI test images and compare to the standard (offline) method.

\addtolength{\tabcolsep}{-\tabL}
{ \renewcommand{\arraystretch}{0.0}
	\begin{figure}
		\centering
		\begin{tabular}{cc}
			\multicolumn{2}{c}{
			\includegraphics[trim={{.15\linewidth} {.09\linewidth} {.095\linewidth} {.12\linewidth}}, clip, width=0.48\linewidth, height = 0.44\linewidth]
			{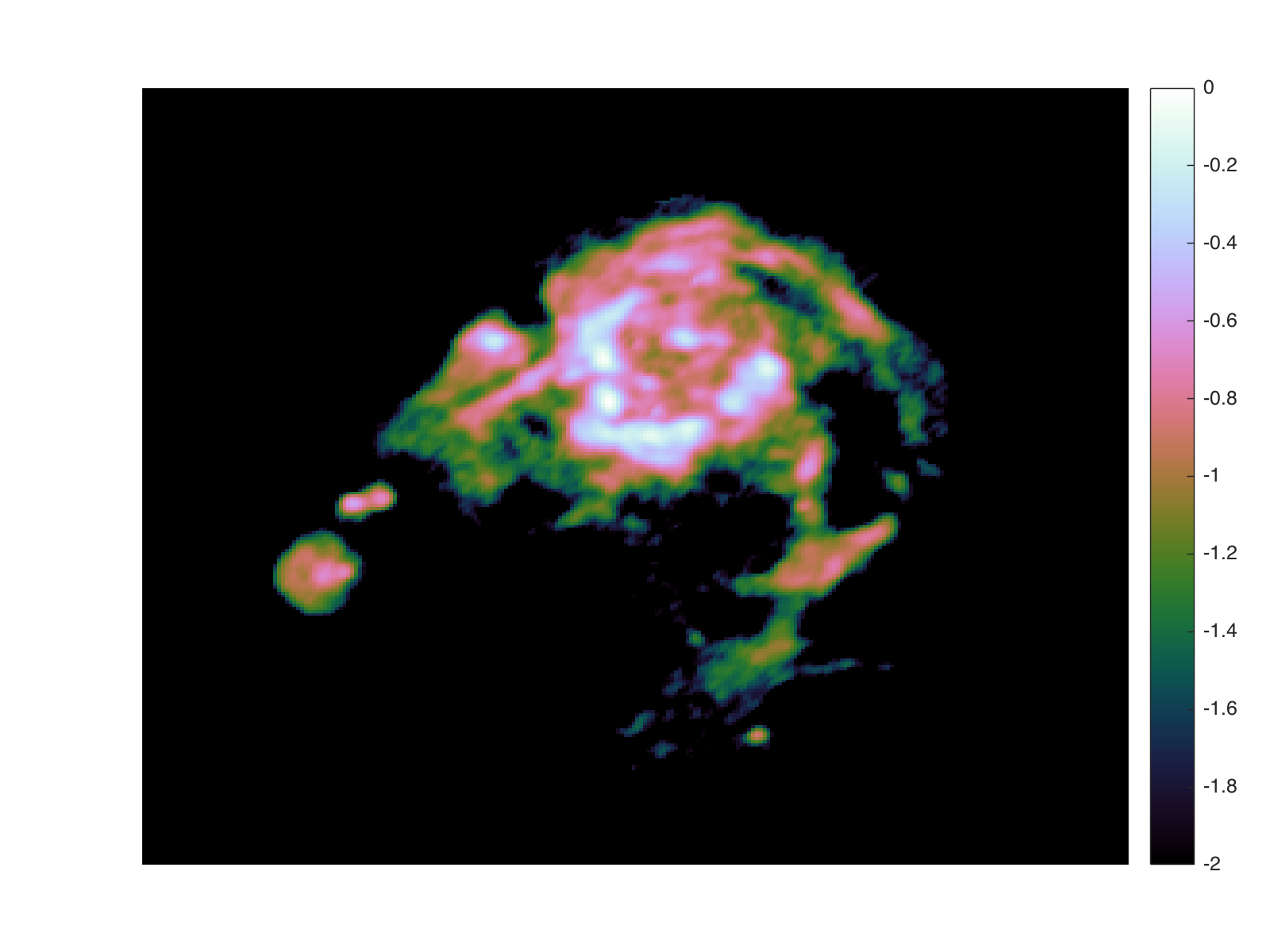} 
			} \vspace{-0.00in}
			\\
			\multicolumn{2}{c}{ (a) M31}
			\\
			\multicolumn{2}{c}{
			\includegraphics[trim={{.15\linewidth} {.09\linewidth} {.095\linewidth} {.12\linewidth}}, clip, width=0.96\linewidth, height = 0.44\linewidth]
			{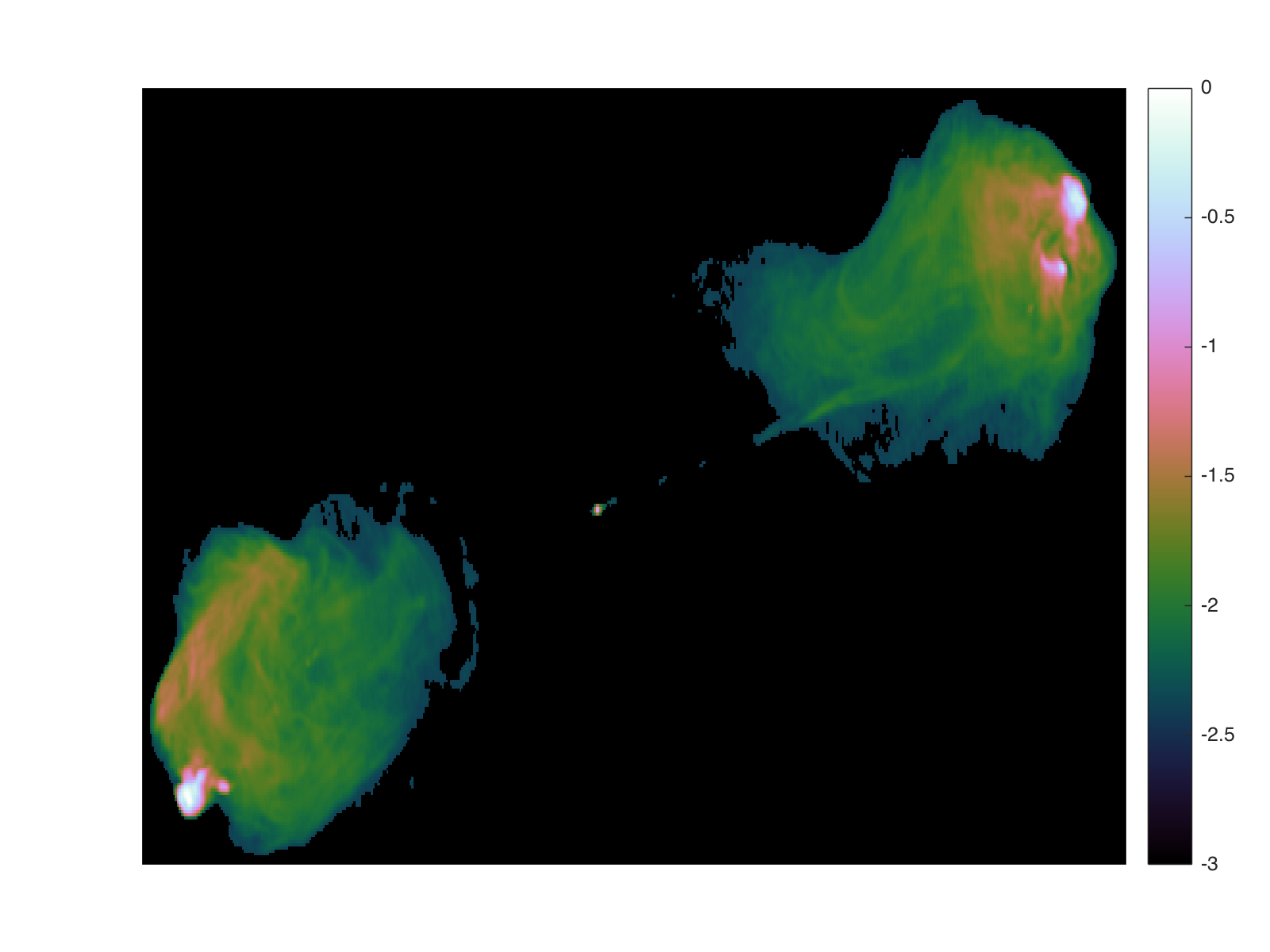} 
			} \vspace{-0.00in}
			\\
			\multicolumn{2}{c}{ (b) Cygnus A }
			\\
			\includegraphics[trim={{.15\linewidth} {.09\linewidth} {.05\linewidth} {.12\linewidth}}, clip, width=0.48\linewidth, height = 0.44\linewidth]
			{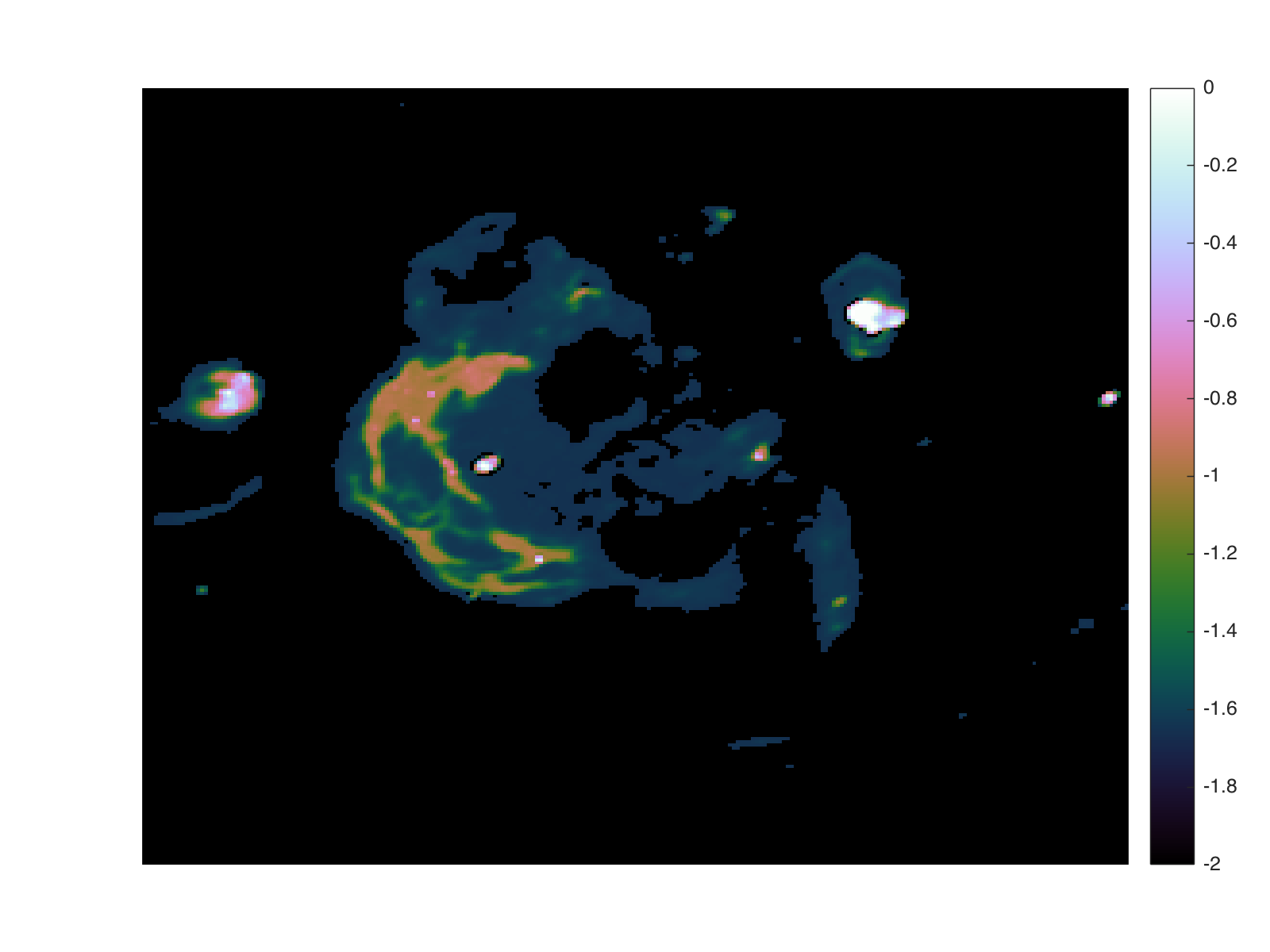} &                     
			\includegraphics[trim={{.15\linewidth} {.09\linewidth} {.05\linewidth} {.12\linewidth}}, clip, width=0.48\linewidth, height = 0.44\linewidth]
			{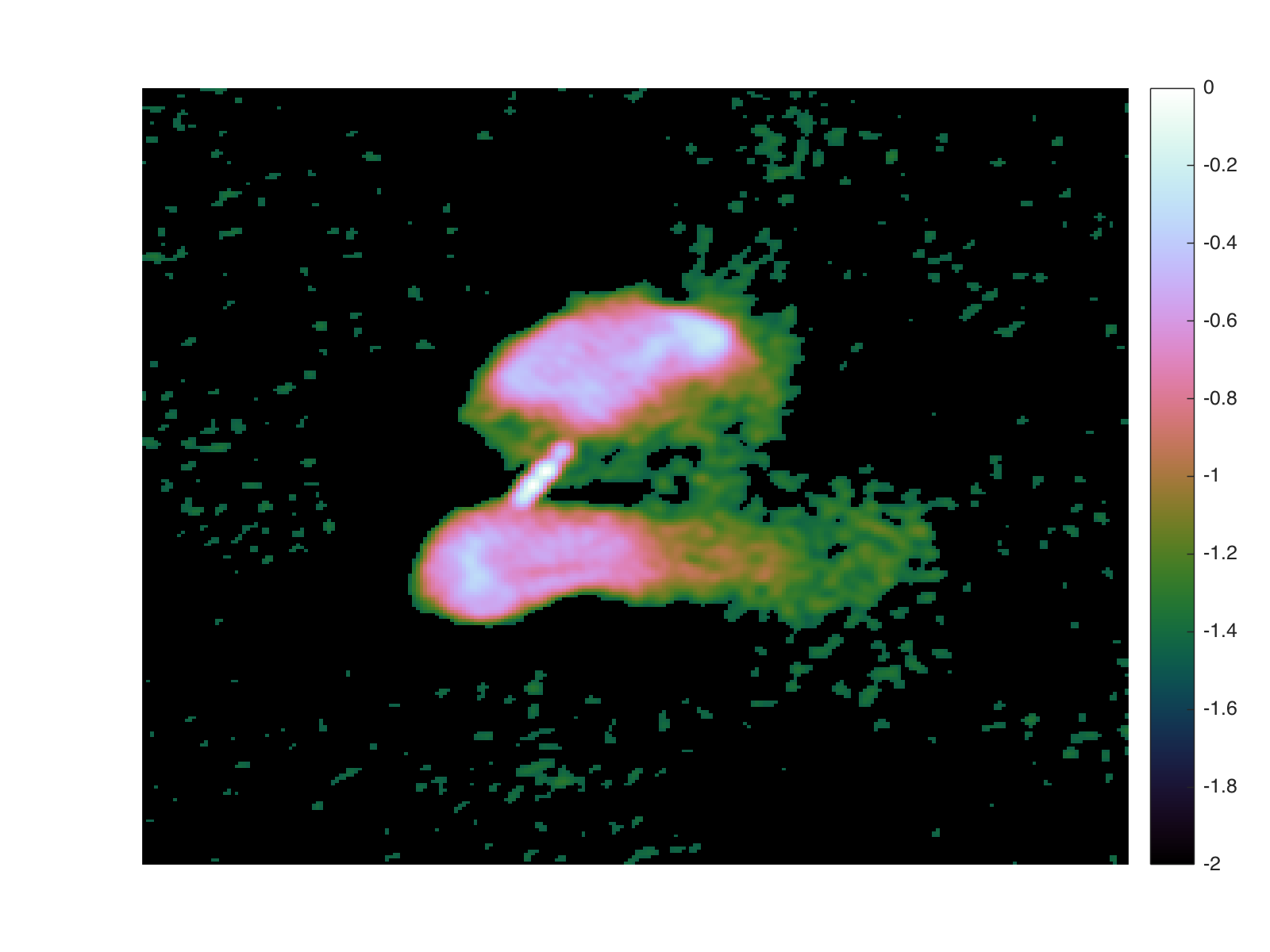} 
			\\
			{\small (c) W28 }      & {\small (d) 3C288 } 
		\end{tabular}
		\caption{Radio images used to test the performance of the online and standard reconstruction methods. 
			Panels (a)--(d): M31, Cygnus A, W28 and 3C288 radio galaxies, respectively. }
		\label{fig-img}
	\end{figure}
}
\addtolength{\tabcolsep}{\tabL}

\subsection{Simulations}
Figure \ref{fig-img} shows the four images used for tests, including an HI region of the M31 galaxy (size $256\times 256$), 
the Cygnus A radio galaxy (size $256\times 512$), the W28 supernova remnant (size $256\times 256$), 
and the 3C288 radio galaxy (size $256\times 256$). The hardware used to perform simulations and 
subsequent numerical experiments is a laptop (Macbook) with CPU of 2.2 GHz, 
four Intel Core i7 processors and memory of 16 GB. All the codes are running on {\sc Matlab} R2015b.

The measurement (sensing) operator, $\bm{\mathsf{\Phi}}_{k}$, $k = 1, \ldots, B$, 
used for simulations, for simplicity, is a Fast Fourier Transform (FFT) operator, $\bm{\mathsf{F}}$,
followed by a masking operation, $\bm{\mathsf{M}}_{k}$, {\it i.e.},
$\bm{\mathsf{\Phi}}_{k} = \bm{\mathsf{M}}_{k}\bm{\mathsf{F}}$.
In principle, this can easily be extended to the more realistic case for measurements that lie off the Fourier grid, 
where $\bm{\mathsf{M}}_{k}$ is replaced with a degridding matrix, and zero padding and degridding correction 
is applied before the the FFT (see \citealt{PMdCOW16} for more details).

The entire visibilities are generated randomly through the variable density sampling profile \citep{PVW11} in half
the Fourier plane with 10\% of discrete Fourier coefficients of each ground truth image. 
Then, the visibilities are corrupted by complex Gaussian noise with zero mean and standard deviation $\sigma$,
where $\sigma = \|f\|_{\infty} 10^{-\textrm{SNR}/20}$, $\|\cdot \|_{\infty}$ is the infinity norm (referring to the maximum absolute value
of components of $f$), and SNR (signal to noise ratio) is set to 30 dB for all simulations. The definition of SNR is given by
\begin{equation}
	{\rm SNR}  = 20 \log_{10} \frac{\|{\vect x}\|_2}{ \|{\vect x} - {\vect x}^* \|_2},
\end{equation}
where ${\vect x}$ is the ground truth of a reconstruction $ {\vect x}^*$.
The basis $\bm{\mathsf{\Psi}}$ in the analysis and synthesis models \eqref{eqn:ir-un-af} and \eqref{eqn:ir-un-sf}
is set to Daubechies 8 wavelets, which can be applied by using the {\sc Matlab} built-in function {\tt wavedec2}.
Note that appreciable difference between the results of 
the analysis and synthesis models is not expected, since this basis satisfies 
$\bm{\mathsf{\Psi}}^\dagger \bm{\mathsf{\Psi}} = \bm{\mathsf{ I}}$. 
Nevertheless, using other varieties ({\it e.g.} overcomplete bases) is straightforward and likely to improve reconstruction fidelity.
Since the focus of this article is the online method rather than optimising imaging performance, we leave the discussion 
of other cases of $\bm{\mathsf{\Psi}}$ for future investigation.

In practice, visibilities will be provided as they are acquired by a telescope and so likely to be observed along $uv$ tracks. 
One simple way to capture this, approximately, in our simplified experimental setup is to take visibilities based on 
distance from the origin. 
In any case, to demonstrate the robustness of the online method with respect to different visibility splitting settings,
we also consider uniformly random visibility selection (from the variable density sampled profile) 
as a visibility splitting rule. 
The performance of the online method with respect to different number of visibility blocks is investigated as well. 

For simplicity, the $\ell_1$ regularisation parameter $\mu$ is fixed to $10^{4}$ by trial-and-error inspection. 
For the standard algorithm, the maximum iteration number $i_{\rm max}$
is used as its stopping criterion, and is set to 50, which is sufficient for reasonably good reconstruction.
Without loss of generality, we report the results corresponding to the analysis model 
and do not show results for the synthesis model. The reason being that the performance difference of 
the tested methods between the synthesis and the analysis models is negligible under an orthogonal basis,
as anticipated.

\addtolength{\tabcolsep}{-\tabL}
\begin{figure*}
	\centering
	\begin{tabular}{ccc}
		\includegraphics[trim={{.15\linewidth} {.07\linewidth} {.02\linewidth} {.072\linewidth}}, clip, width=0.32\linewidth, height = 0.28\linewidth]
		{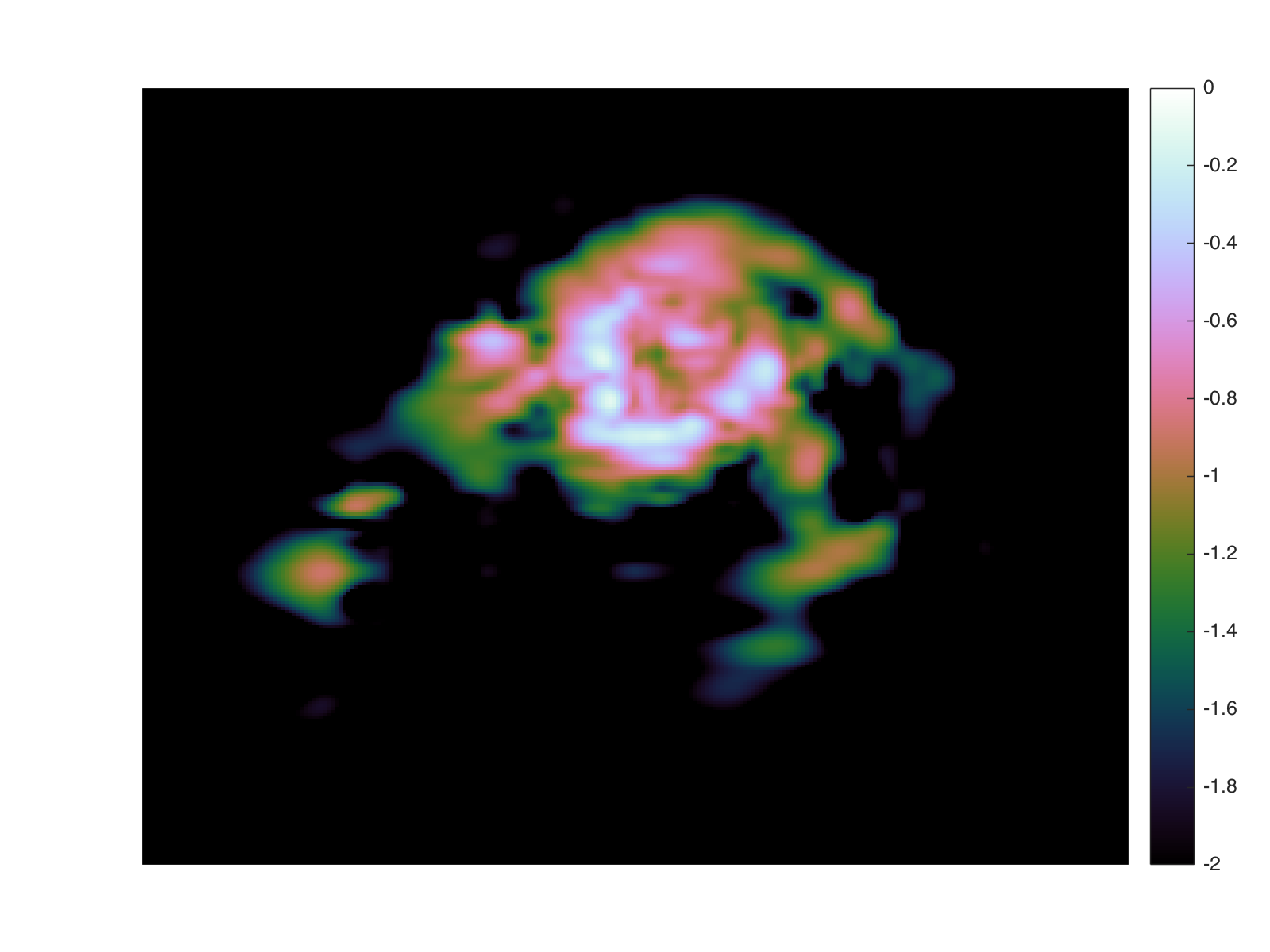}  \put(-170,60){\rotatebox{90}{ M31}}   &
		\includegraphics[trim={{.15\linewidth} {.07\linewidth} {.02\linewidth} {.072\linewidth}}, clip, width=0.32\linewidth, height = 0.28\linewidth]
		{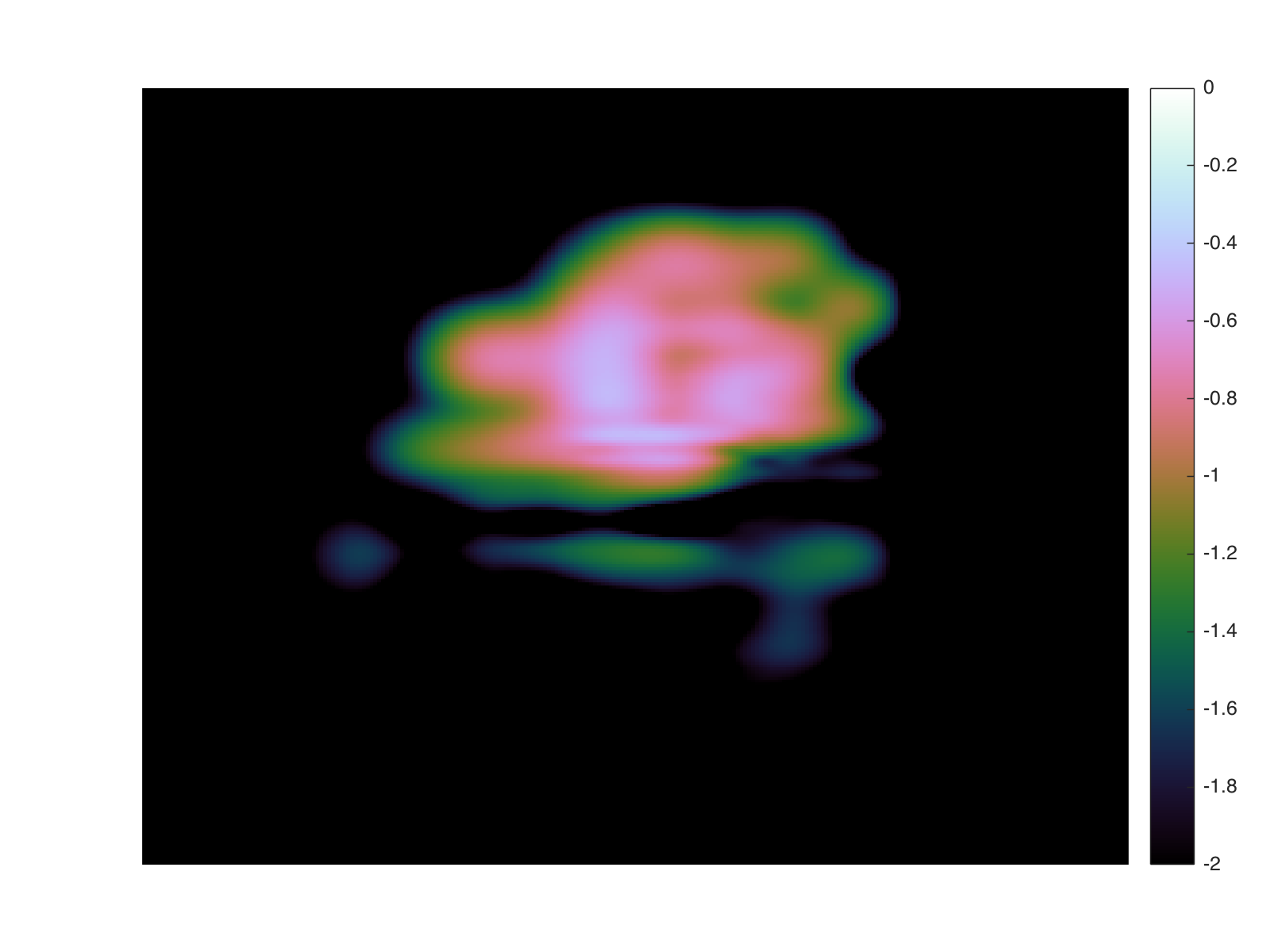} &
		\includegraphics[trim={{.15\linewidth} {.07\linewidth} {.02\linewidth} {.072\linewidth}}, clip, width=0.32\linewidth, height = 0.28\linewidth]
		{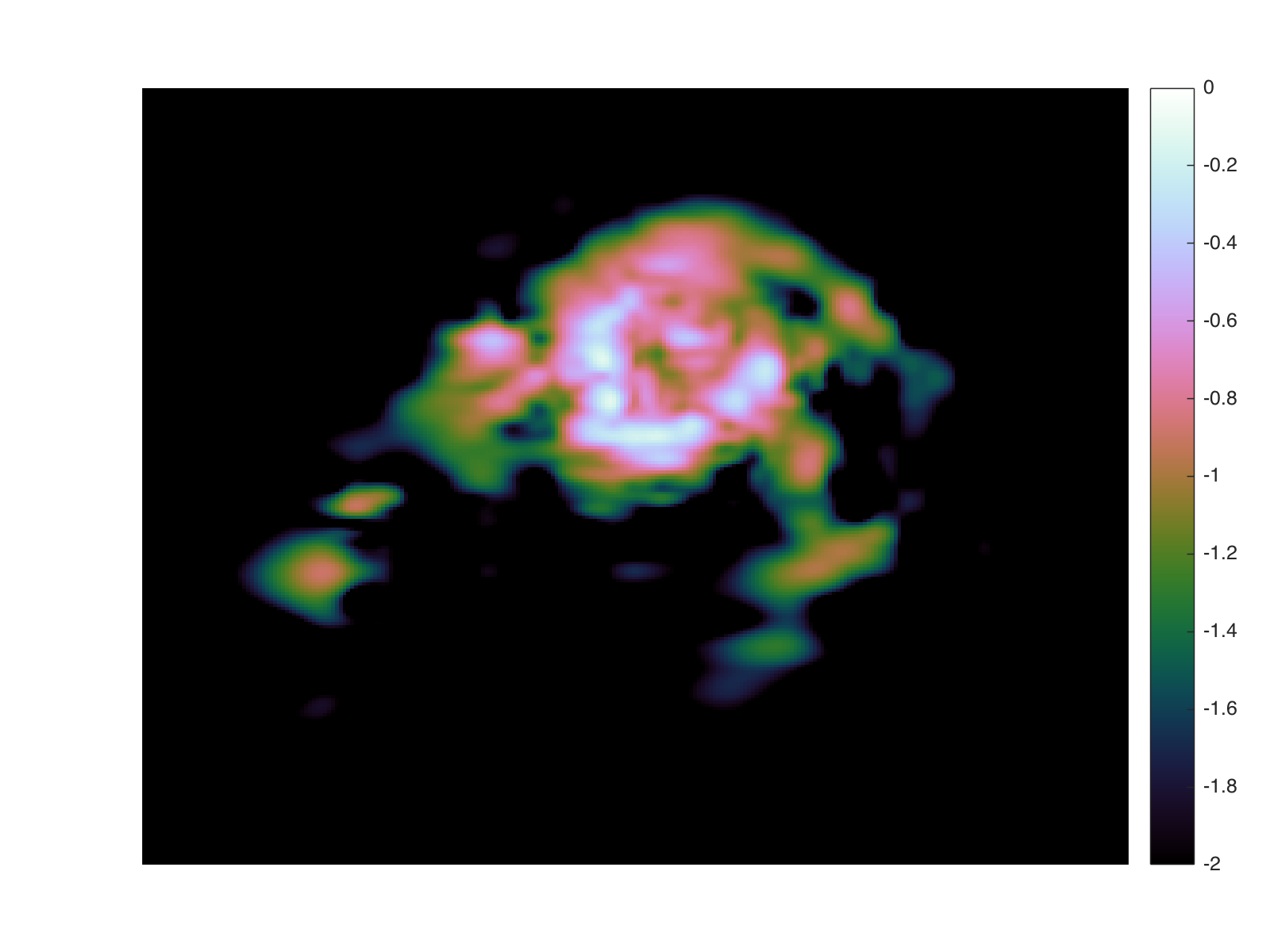} 
		\\
		\includegraphics[trim={{.15\linewidth} {.07\linewidth} {.02\linewidth} {.072\linewidth}}, clip, width=0.32\linewidth, height = 0.16\linewidth]
		{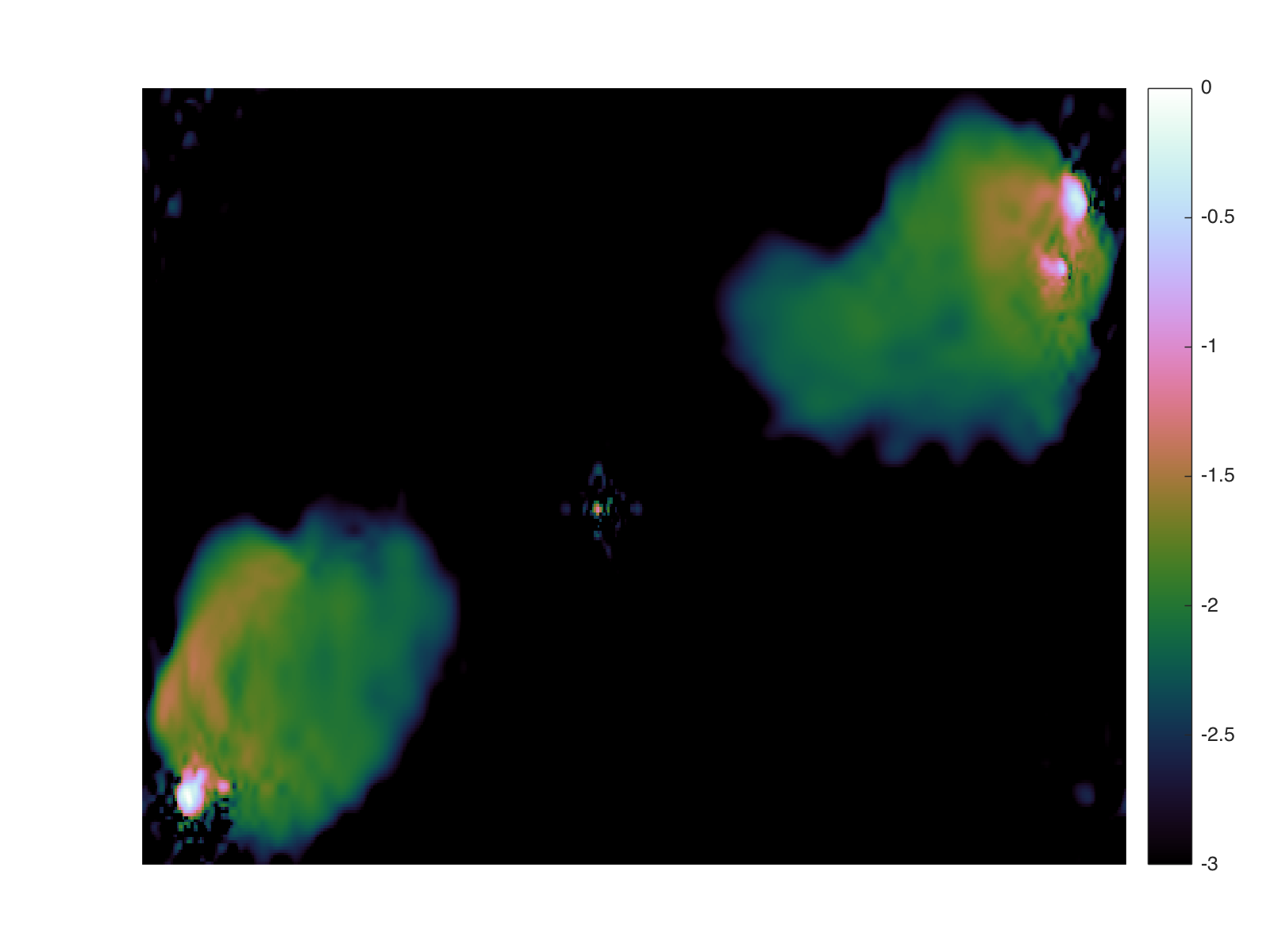}   \put(-170,20){\rotatebox{90}{ Cygnus A}} &
		\includegraphics[trim={{.15\linewidth} {.07\linewidth} {.02\linewidth} {.072\linewidth}}, clip, width=0.32\linewidth, height = 0.16\linewidth]
		{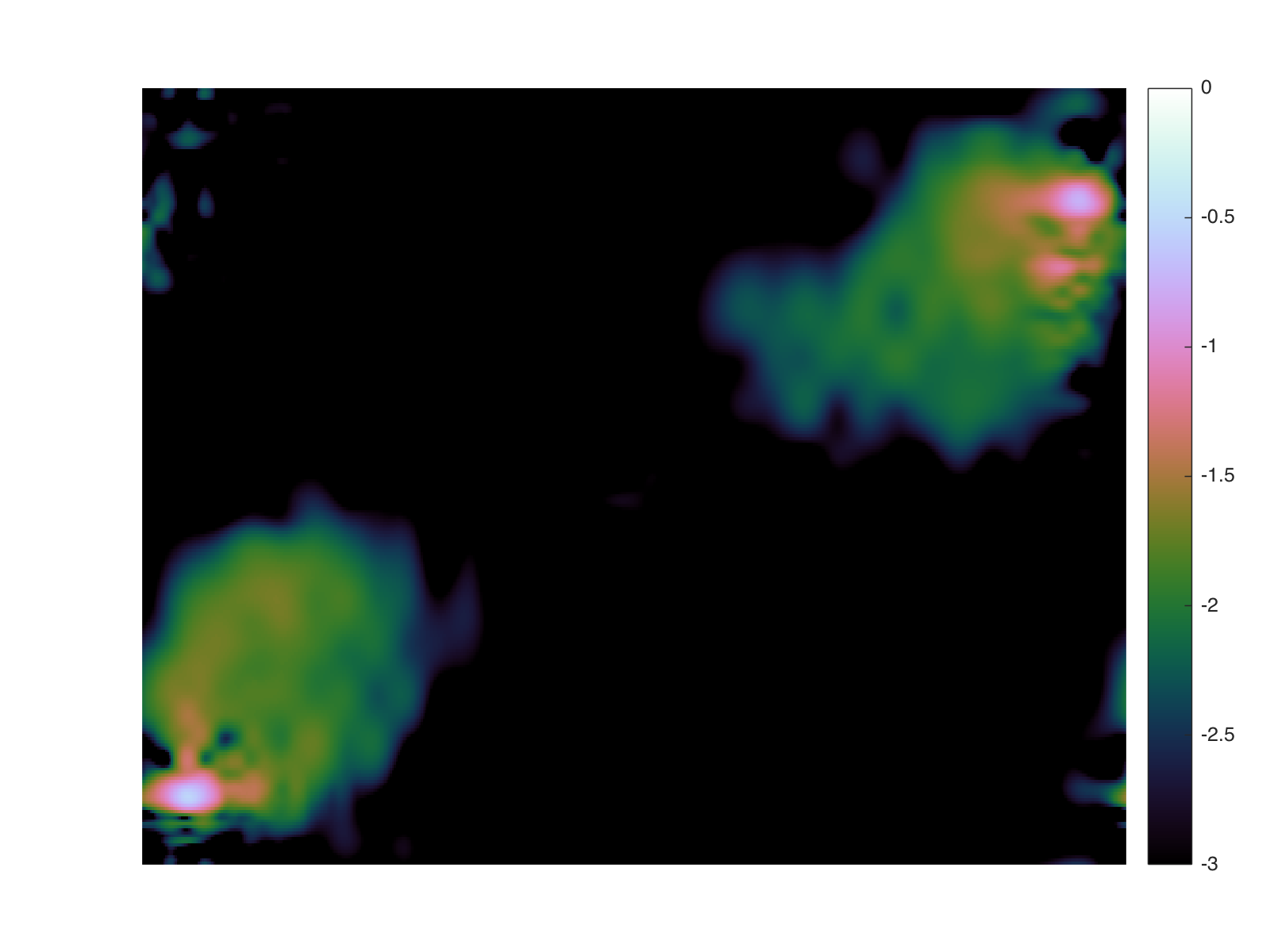} &
		\includegraphics[trim={{.15\linewidth} {.07\linewidth} {.02\linewidth} {.072\linewidth}}, clip, width=0.32\linewidth, height = 0.16\linewidth]
		{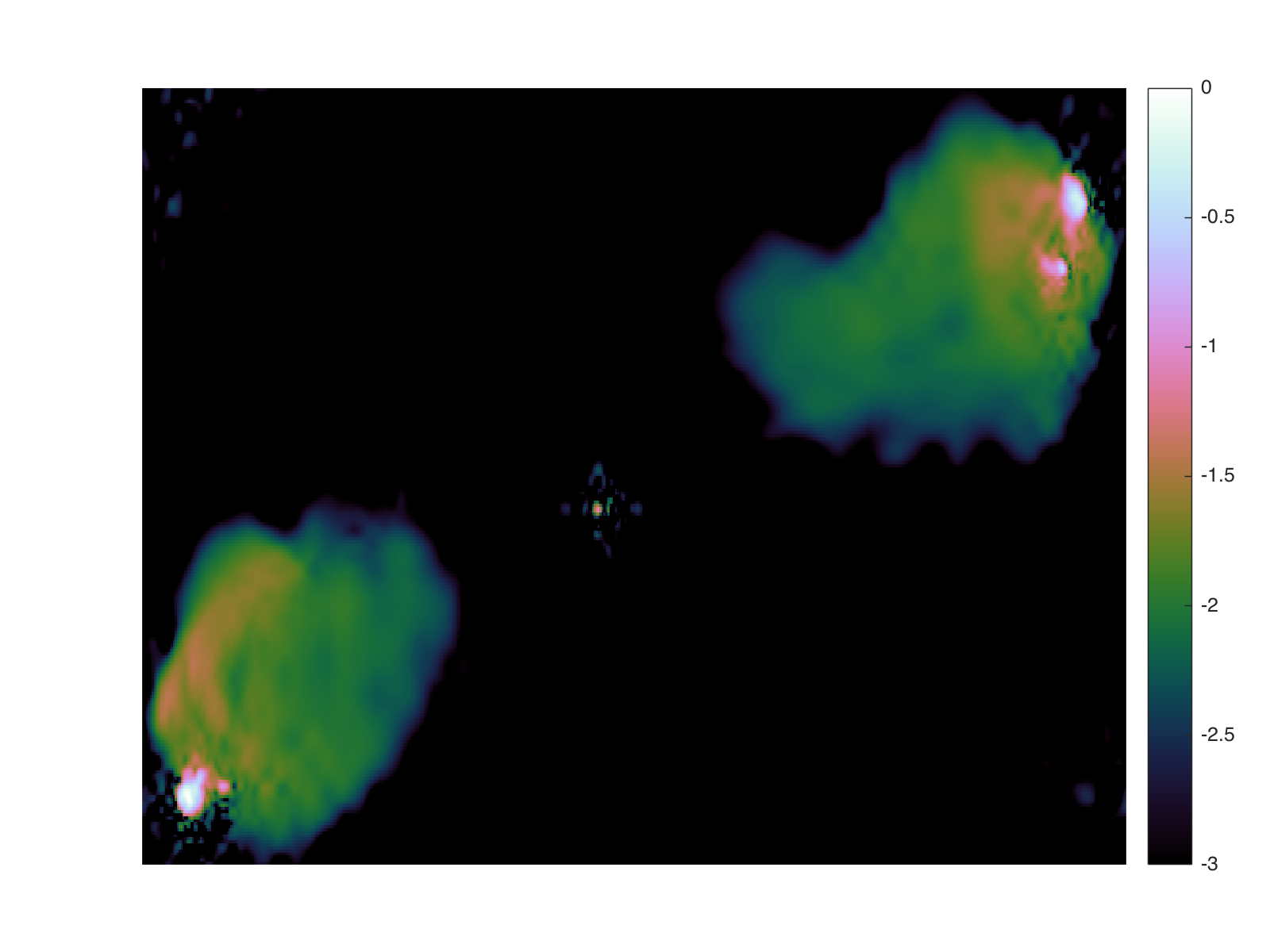} 
		\\
		\includegraphics[trim={{.15\linewidth} {.07\linewidth} {.02\linewidth} {.072\linewidth}}, clip, width=0.32\linewidth, height = 0.28\linewidth]
		{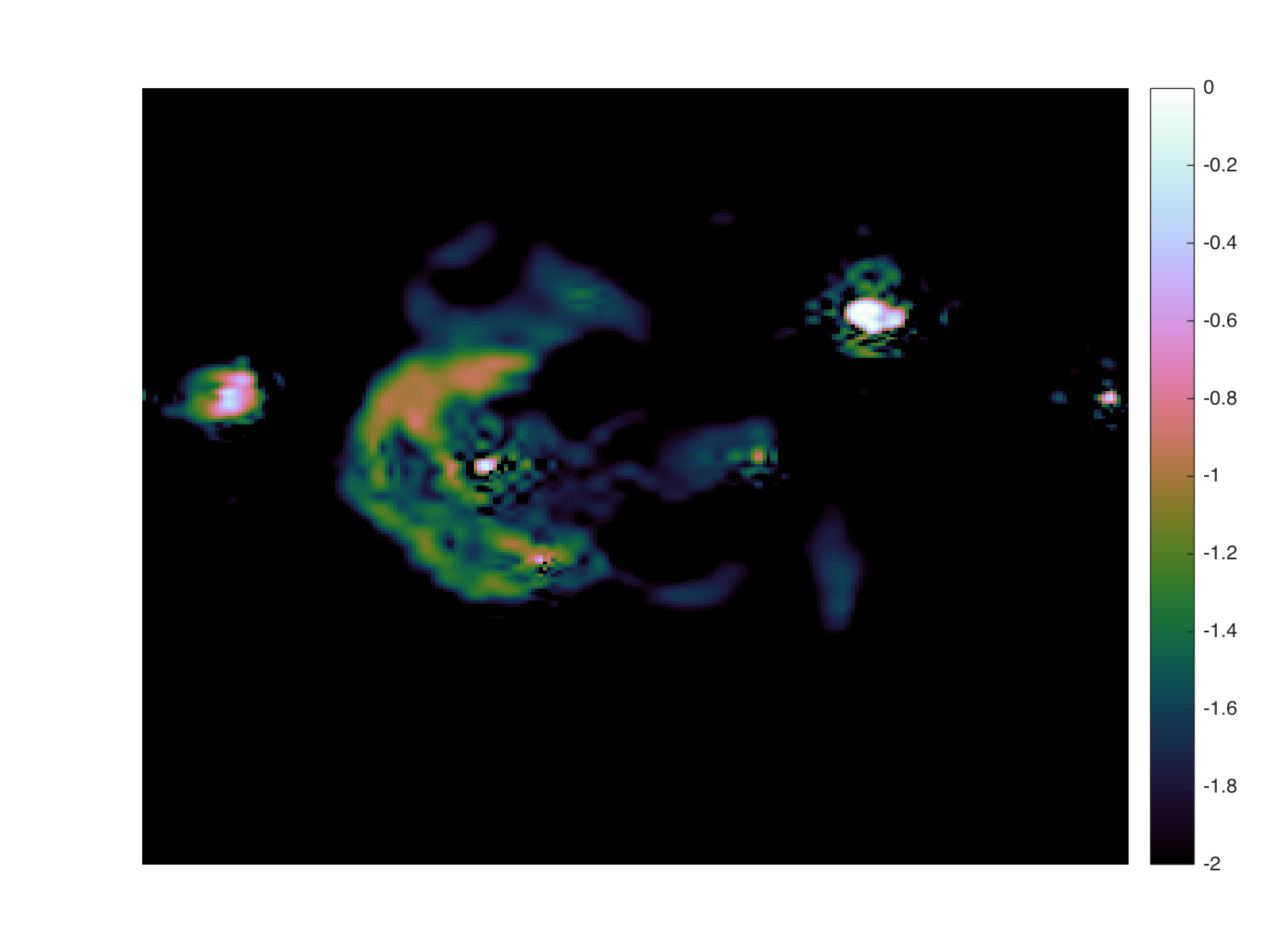}  \put(-170,60){\rotatebox{90}{ W28}} &
		\includegraphics[trim={{.15\linewidth} {.07\linewidth} {.02\linewidth} {.072\linewidth}}, clip, width=0.32\linewidth, height = 0.28\linewidth]
		{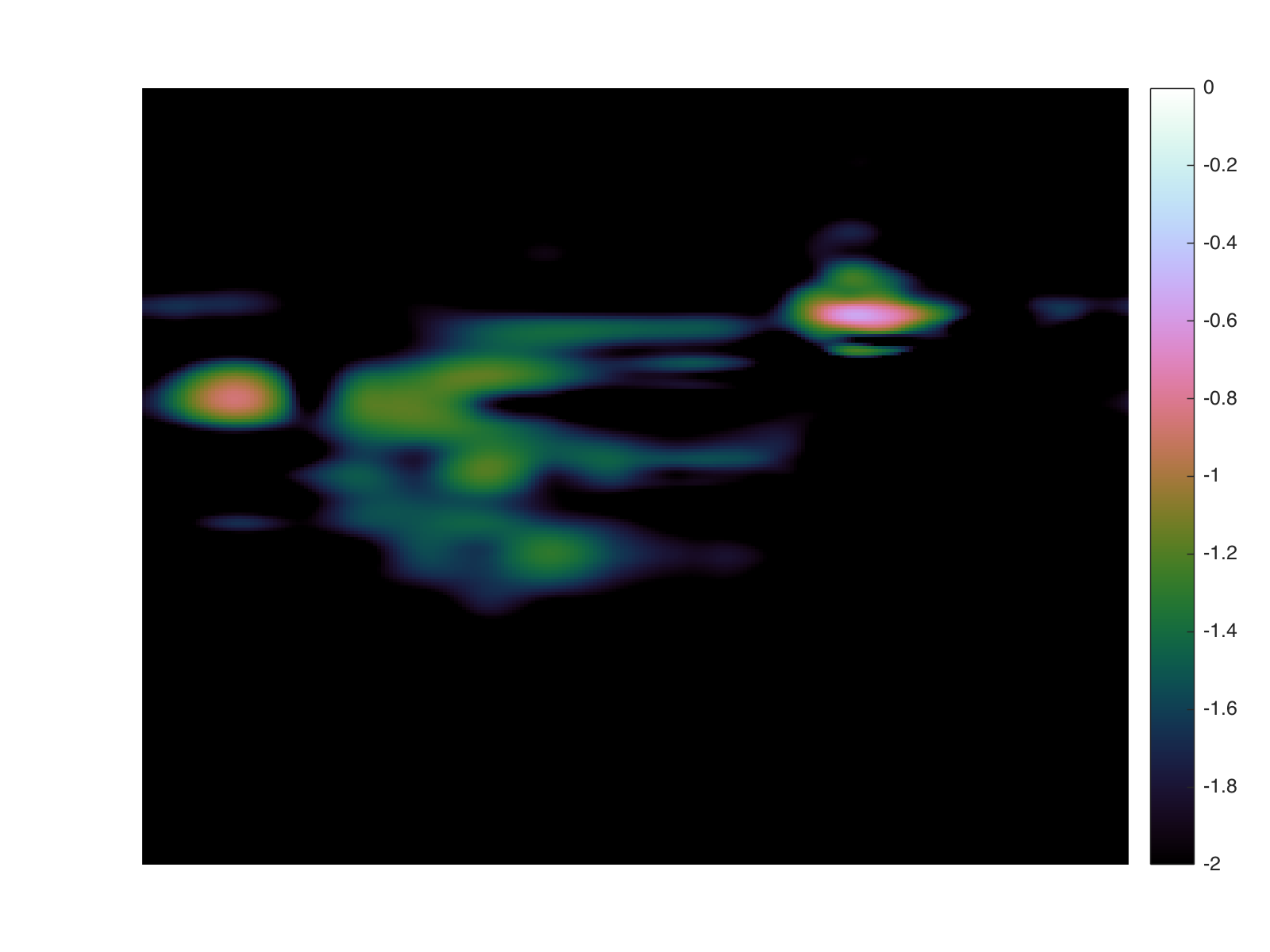} &
		\includegraphics[trim={{.15\linewidth} {.07\linewidth} {.02\linewidth} {.072\linewidth}}, clip, width=0.32\linewidth, height = 0.28\linewidth]
		{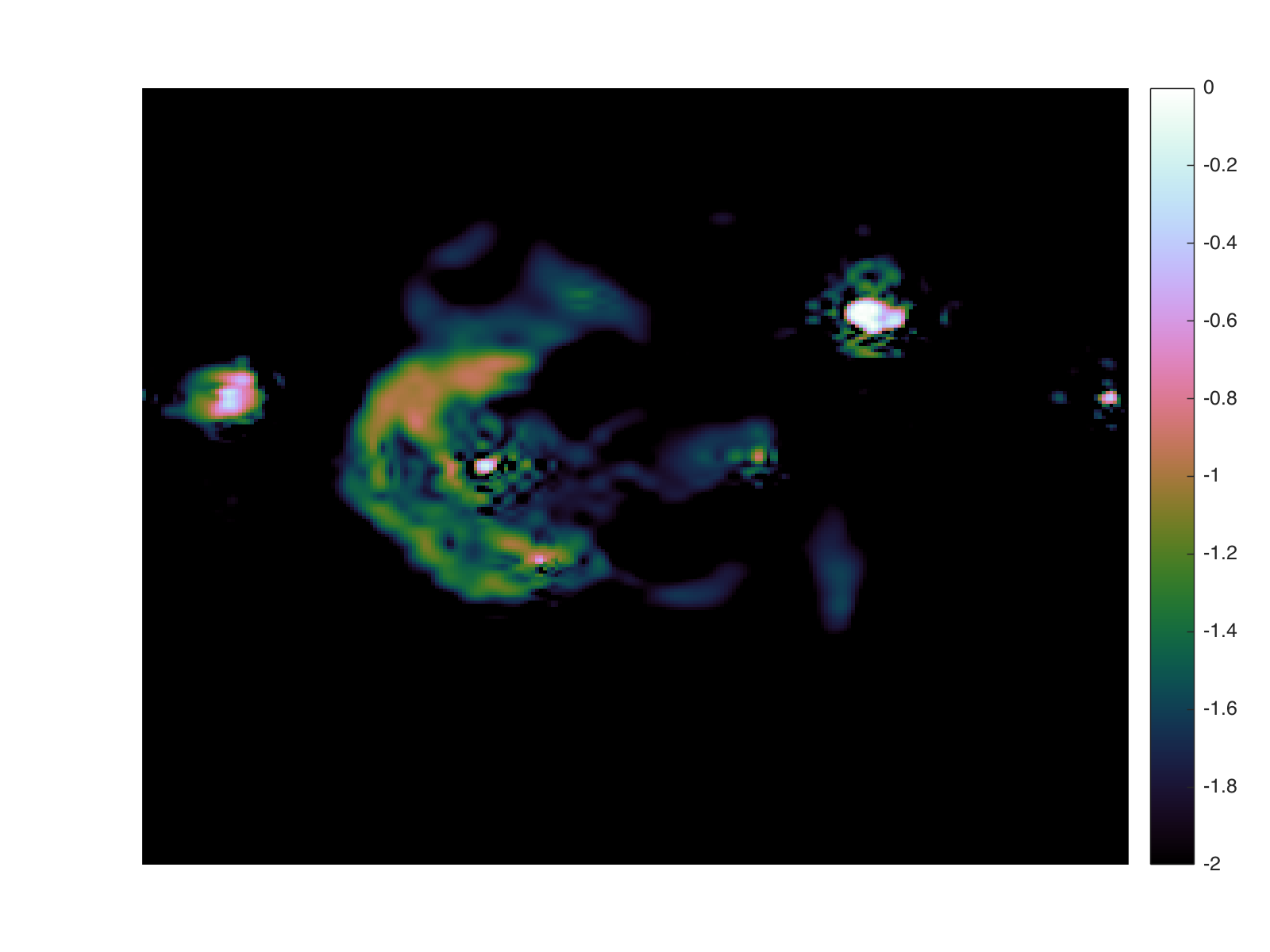} 
		\\
		\includegraphics[trim={{.15\linewidth} {.07\linewidth} {.02\linewidth} {.072\linewidth}}, clip, width=0.32\linewidth, height = 0.28\linewidth]
		{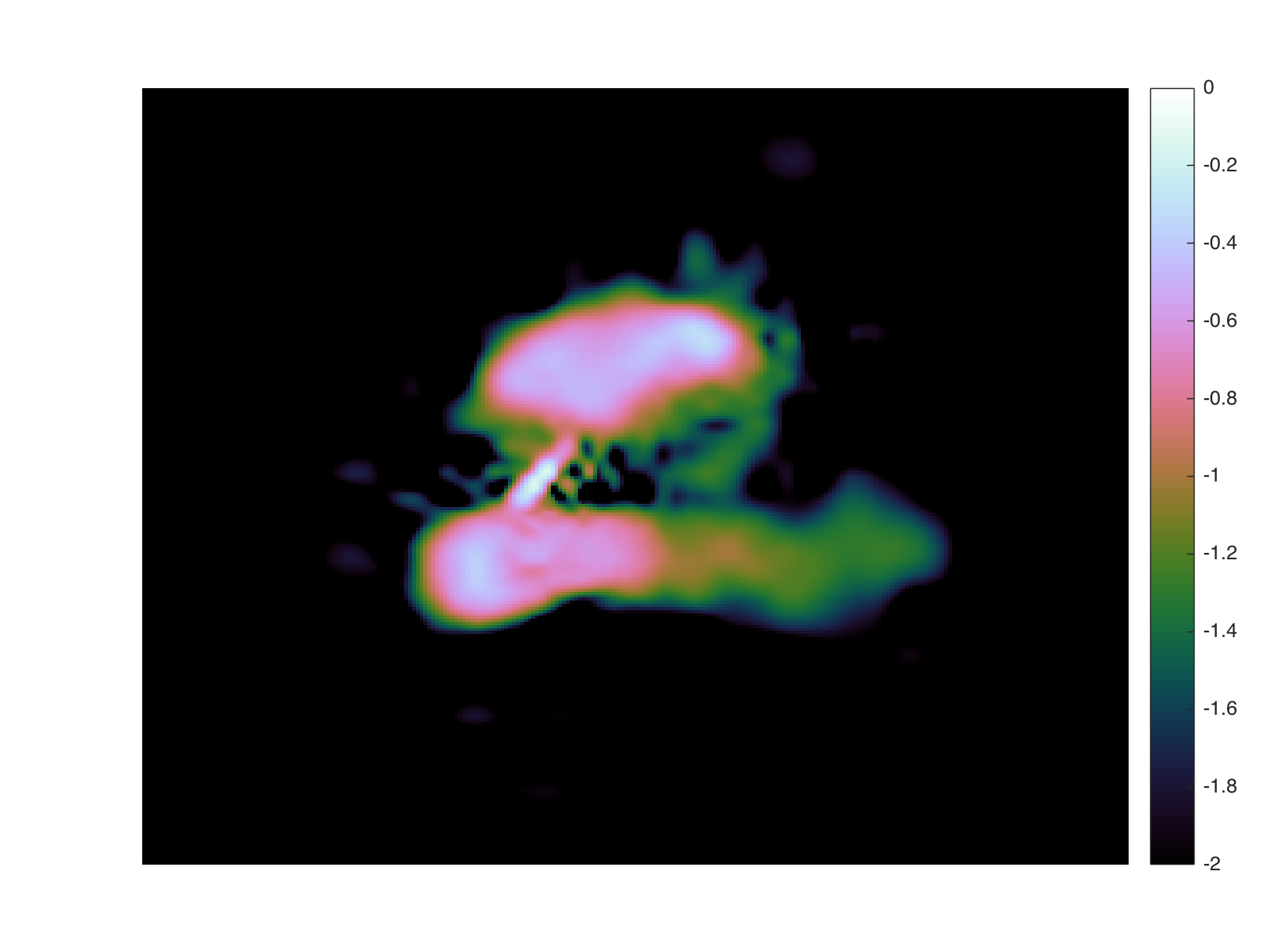} \put(-170,57){\rotatebox{90}{ 3C288}} &
		\includegraphics[trim={{.15\linewidth} {.07\linewidth} {.02\linewidth} {.072\linewidth}}, clip, width=0.32\linewidth, height = 0.28\linewidth]
		{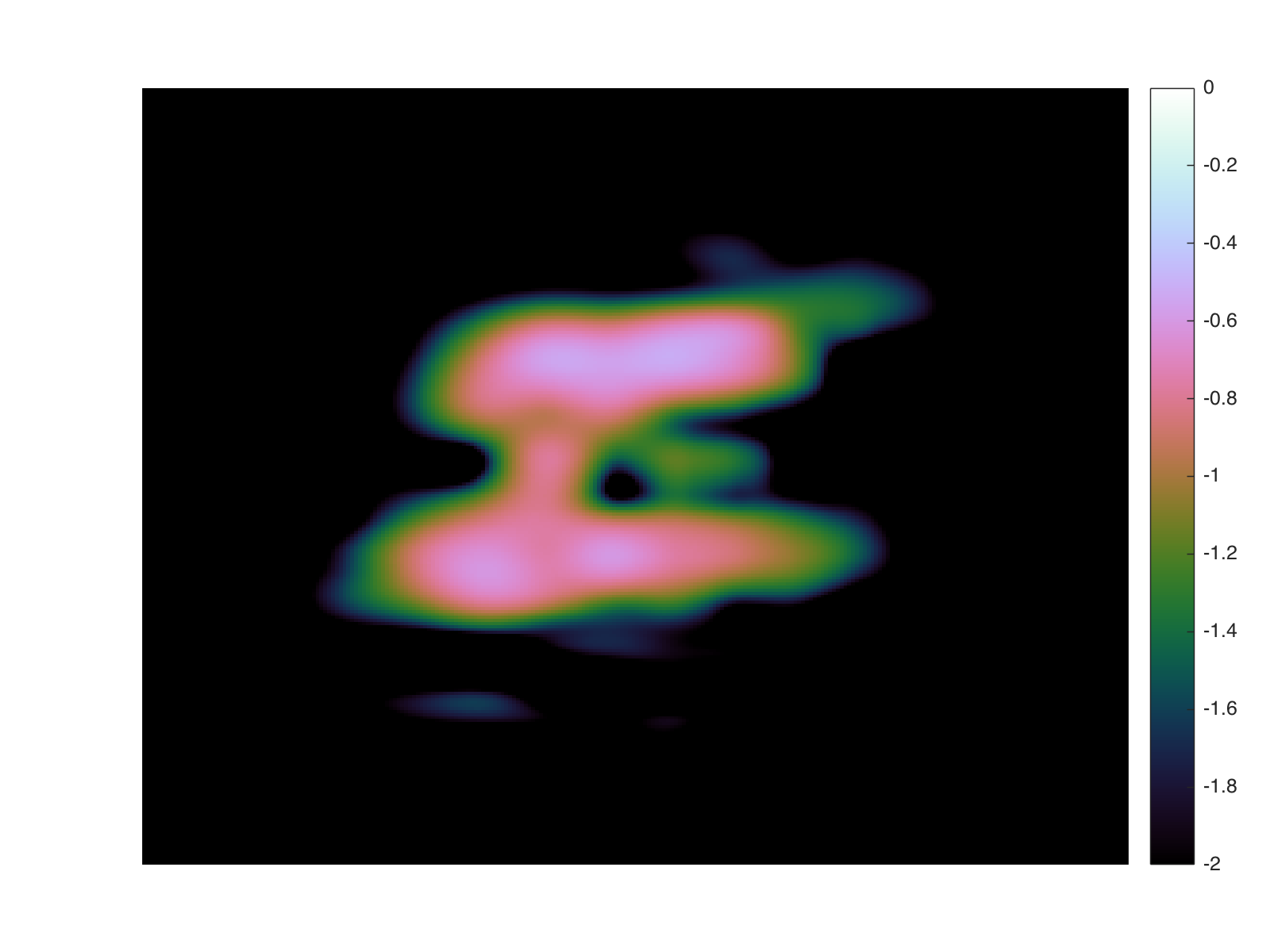} &
		\includegraphics[trim={{.15\linewidth} {.07\linewidth} {.02\linewidth} {.072\linewidth}}, clip, width=0.32\linewidth, height = 0.28\linewidth]
		{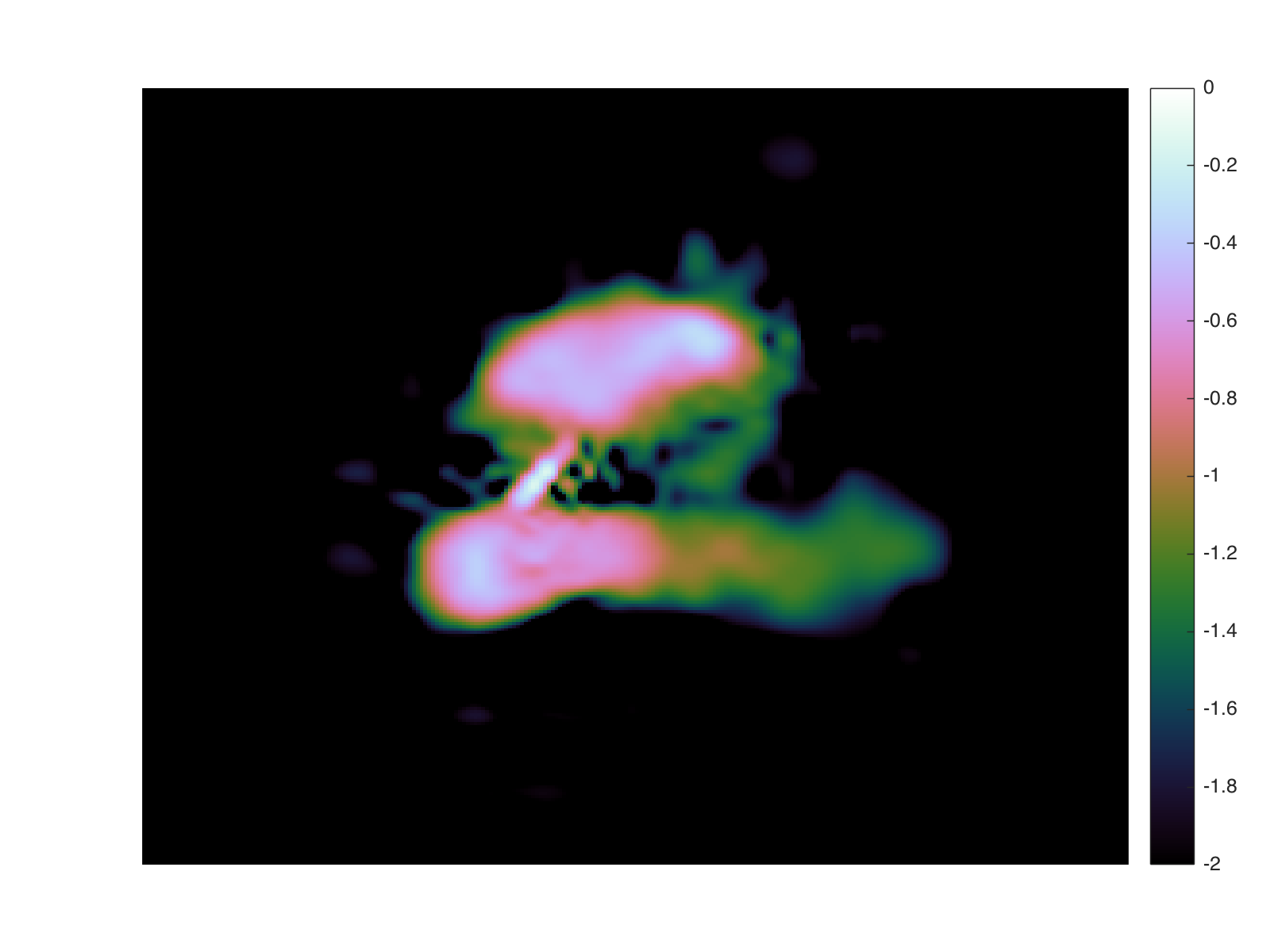} 
		\\
		{ (a) Offline algorithm} & {(b) Offline algorithm } & { (c) Online algorithm (our work)}                    \\
		(storage: 100\% visibilities)         & (storage: 2\% visibilities)        & { (storage: 2\% visibilities)} 
	\end{tabular}
	\caption{Image reconstruction results of the standard (offline) algorithm and the online algorithm (our work)
		for images M31 (first row), Cygnus A (second row), W28 (third row), and 3C288 (fourth row). The number of iterations for the
		tested methods is set to 50, and all images are shown in ${\tt log}_{10}$ scale. 
		Panels (a) and (b): results of the standard algorithm correspond to 100\%  and 2\% of all acquired visibilities, respectively, 
		where 10\% of discrete visibilities are acquired in our simple discrete simulations.
		Panel (c): result of the online algorithm, with visibilities gradually increased from 2\% to 100\% block-by-block
		(each visibility block contains 2\% of all acquired visibilities),
		according to the distance of the visibilities to the origin of the Fourier domain.
    Clearly, when using all visibilities similar reconstruction quality is obtained by both the standard offline and online methods (panels (a) and (c), respectively), 
    which achieve the same SNR (14.2946 dB) for M31, for example. However, the online method requires storage for only 2\% of the acquired  visibilities, 
    whereas the offline method must store them all.  Panel (b) shows reconstructions when using the the same amount of storage for the standard offline method 
    as required by the online method.  In this setting there are too few visibilities to produce a reasonable reconstruction.
		We emphasise that the online algorithm combines the reconstruction task with the 
		visibility acquisition stage, which can significantly improve the reconstruction speed and dramatically reduce storage costs.
	}
	\label{fig-m31}
\end{figure*}
\addtolength{\tabcolsep}{\tabL}

\subsection{Algorithm performance}
We present the results of the tested methods  -- the online method (our work) and the standard method -- 
and their comparison in terms of reconstruction quality, visibility storage requirements and computational cost. 
Moreover, the quantitative performance of the methods is also analysed with respect to different 
settings regarding the number of visibility blocks.

\subsubsection{Reconstruction}
Figure \ref{fig-m31} shows the reconstruction results of the tested methods for the analysis model \eqref{eqn:ir-un-af} on the four
test images. For the online method, all the acquired visibilities (10\% of Fourier coefficients) are partitioned into $B = 50$ blocks, where
every block has the same number of visibilities (2\% of the acquired visibilities or 0.2\% of the total number of discrete Fourier coefficients).   
Figure \ref{fig-m31} (a) shows the results of the standard method using the entire observed visibilities.
Similarly, Figure \ref{fig-m31} (b) shows the results of the standard method but with just 
0.2\% of Fourier coefficients uniformly randomly selected from the variable density samples (rather than by distance to the origin), 
which equals the number of Fourier coefficients contained in a single visibility block used for the online method.
This is to test both methods under the situation of limited storage -- a storage which is not large enough to store the entire observed visibilities.
Recall that the online method in principle works for arbitrarily small storage.
Clearly, using all visibilities a good reconstruction is obtained by the standard method, which is not the case when using just 0.2\% of
Fourier coefficients, which is too few to produce a reasonable reconstruction. 
Figure \ref{fig-m31} (c) shows the results of our online method, which are as good as those of the standard method ({\it cf.} Figure \ref{fig-m31} (a))  
under visual validation and quantitative comparison in SNR.
For instance, both methods achieve the same SNR, 14.2946 dB, for test image M31. More detailed quantitative comparison 
of SNR will be deferred to the next subsection. On the whole, the results between each test image are consistent with each other, 
demonstrating the excellent performance of the online method, which
provides approximately as good reconstruction quality as the standard method.

We test that the online method is independent of the visibility splitting strategies, {\it i.e.}, that different strategies produce consistent results. 
As mentioned before, the alternative strategy to split visibilities is to sample the full set of variable density sampled visibilities in a uniformly random manner. 
As an example, the result of M31 shown in Figure \ref{fig-m31} (c, top) is achieved with SNR 14.2946 dB, and an almost identical result with 
SNR 14.2943 dB is obtained using the alternative visibility splitting strategy. For all the tests, we do not show the reconstructed images
when using the different splitting strategy to avoid repetition, since, from visual validation the results are very similar to those shown 
in Figure \ref{fig-m31} (c) based on distance from the origin. 
Seeing this fact, in the following experiments, the splitting strategy based on distance from the origin
is adopted.

\addtolength{\tabcolsep}{-\tabL}
\begin{figure*}
	\centering
	\begin{tabular}{cc}
		\includegraphics[trim={{.07\linewidth} {.03\linewidth} {.08\linewidth} {.05\linewidth}}, clip, width=0.48\linewidth, height = 0.38\linewidth]
		{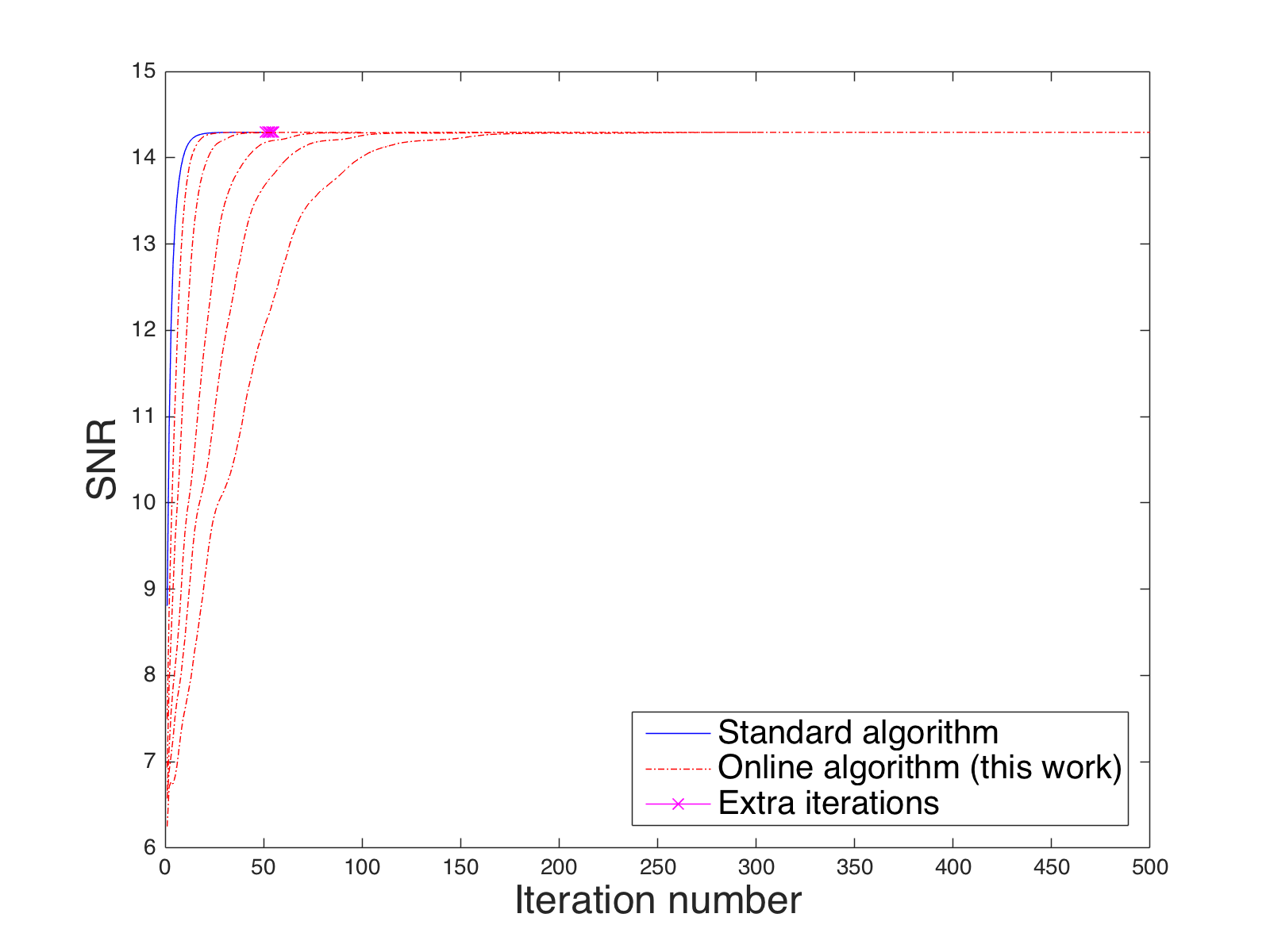}
		\put(-147,145){Visibility storage}
		\put(-147,138){requirements:} \put(-147,126){ 100\%} \put(-148,128){ \vector(-2,1){73}} 
		\put(-147,116){ 2\%}  \put(-148,118){ \vector(-2,1){72}} 
		\put(-147,106){ 1\%}  \put(-148,108){ \vector(-2,1){71}} 
		\put(-147,96){ 0.5\%} \put(-148,98){ \vector(-2,1){67}} 
		\put(-147,86){ 0.3\%} \put(-148,88){ \vector(-2,1){65}} 
		\put(-147,76){ 0.2\%} \put(-148,78){ \vector(-2,1){60}} 
		\put(-216,166){\framebox(25,15){ }}
		            &                  
		\includegraphics[trim={{.07\linewidth} {.03\linewidth} {.08\linewidth} {.05\linewidth}}, clip, width=0.48\linewidth, height = 0.38\linewidth]
		{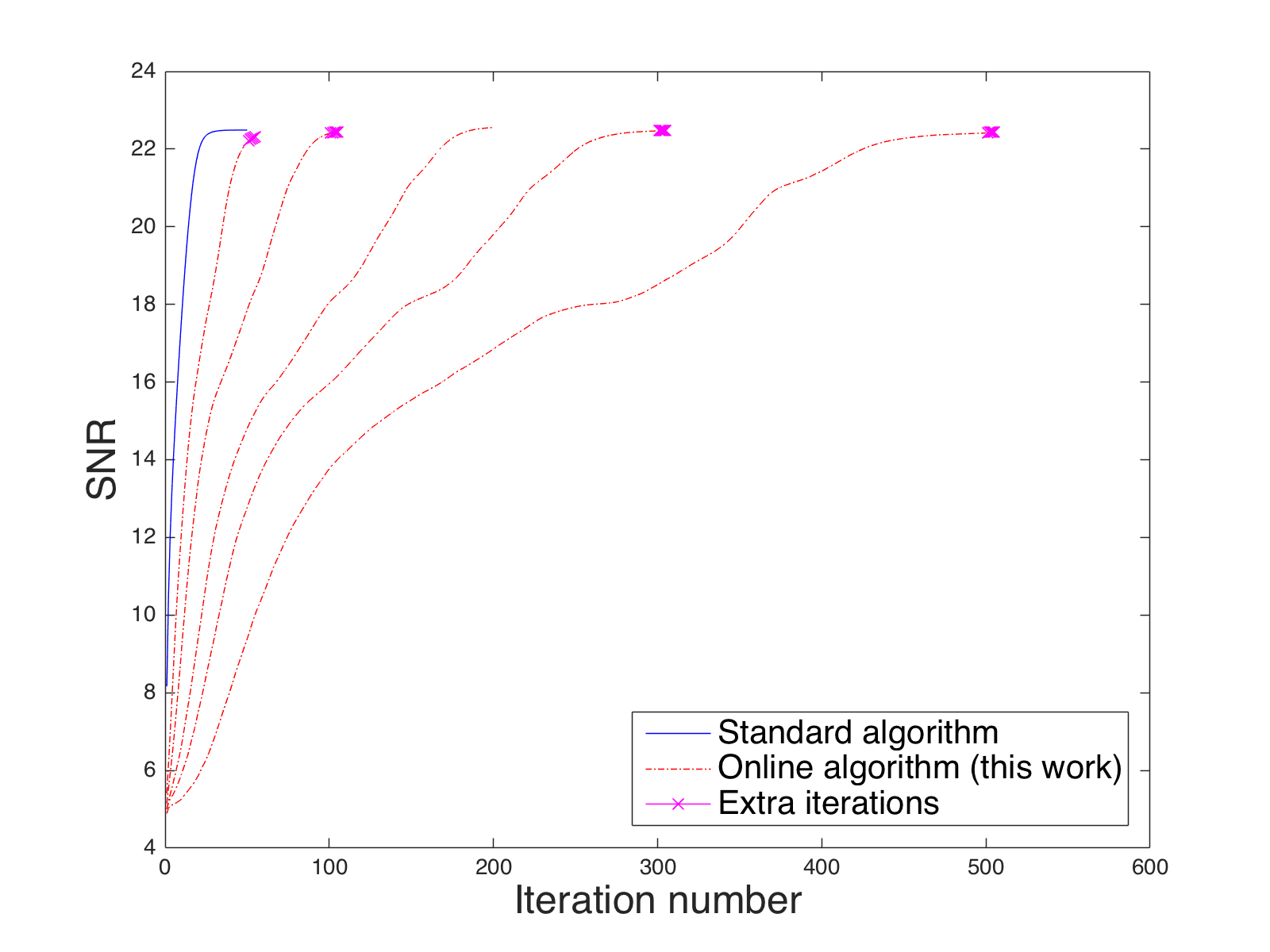} 
		\put(-140,124){Visibility storage}
		\put(-140,117){requirements:} 
		\put(-140,105){ 100\%} \put(-141,107){ \vector(-2,1){78}} 
		\put(-140,95){ 2\%}  \put(-141,97){ \vector(-2,1){73}} 
		\put(-140,85){ 1\%}  \put(-141,87){ \vector(-2,1){69}} 
		\put(-140,75){ 0.5\%} \put(-141,77){ \vector(-2,1){64}} 
		\put(-140,65){ 0.3\%} \put(-141,67){ \vector(-2,1){62}} 
		\put(-140,55){ 0.2\%} \put(-141,57){ \vector(-2,1){56}} 
		\put(-212,160){\framebox(31,25){ }}
		\\
		{ (a) M31 } & { (b) Cygnus A } \\
		\includegraphics[trim={{.07\linewidth} {.03\linewidth} {.08\linewidth} {.05\linewidth}}, clip, width=0.48\linewidth, height = 0.38\linewidth]
		{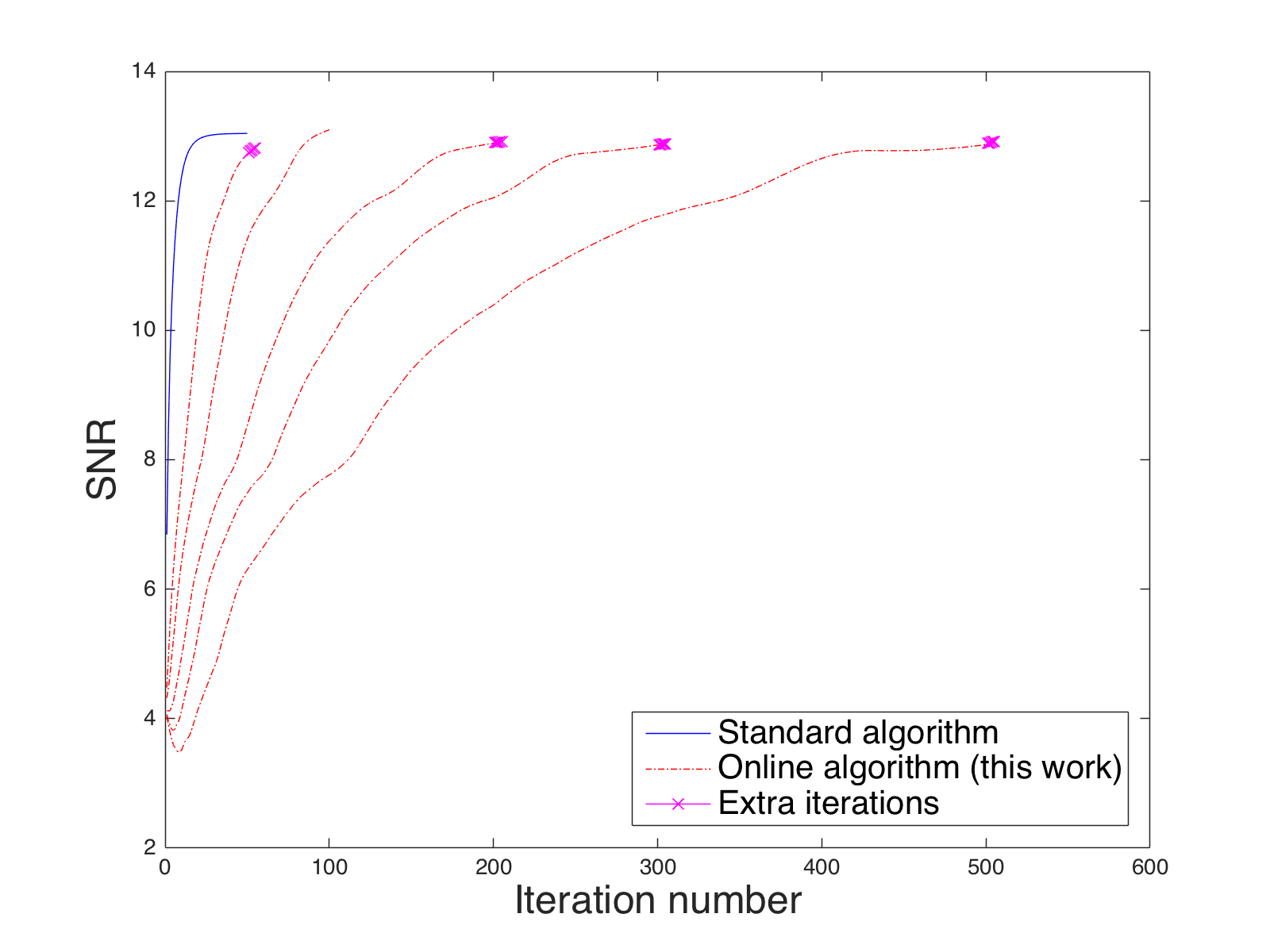}  
		\put(-140,129){Visibility storage}
		\put(-140,122){requirements:} \put(-140,110){ 100\%} \put(-141,112){ \vector(-2,1){81}} 
		\put(-140,100){ 2\%}  \put(-141,102){ \vector(-2,1){75}} 
		\put(-140,90){ 1\%}  \put(-141,92){ \vector(-2,1){70}} 
		\put(-140,80){ 0.5\%} \put(-141,82){ \vector(-2,1){64}} 
		\put(-140,70){ 0.3\%} \put(-141,72){ \vector(-2,1){59}} 
		\put(-140,60){ 0.2\%} \put(-141,62){ \vector(-2,1){56}} 
		\put(-215,160){\framebox(31,25){ }}
		            &                  
		\includegraphics[trim={{.07\linewidth} {.03\linewidth} {.08\linewidth} {.05\linewidth}}, clip, width=0.48\linewidth, height = 0.38\linewidth]
		{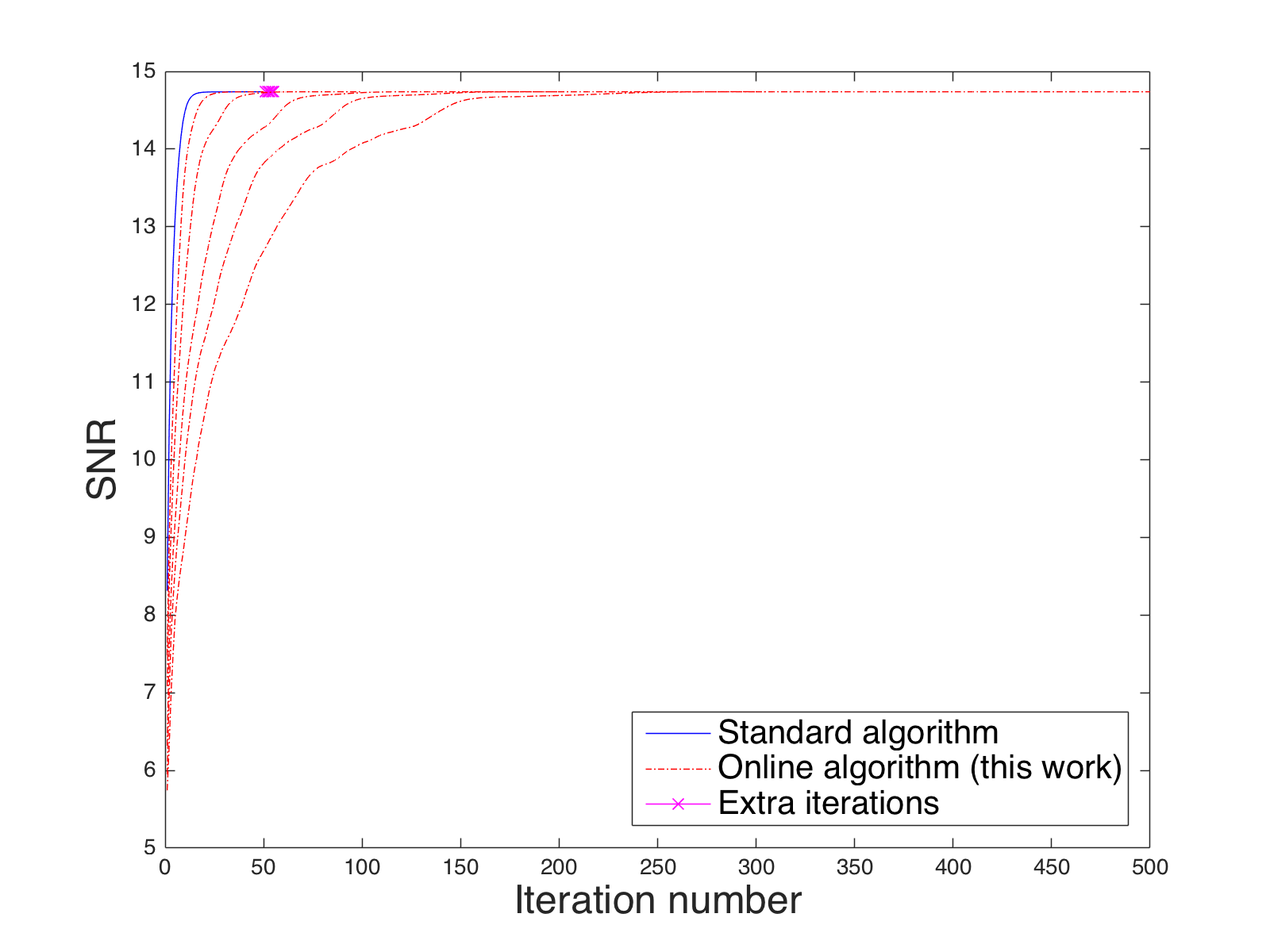} 
		\put(-140,144){Visibility storage} 
		\put(-140,137){requirements:} 
		\put(-140,125){ 100\%} \put(-141,127){ \vector(-2,1){80.5}} 
		\put(-140,115){ 2\%}  \put(-141,117){ \vector(-2,1){80}} 
		\put(-140,105){ 1\%}  \put(-141,107){ \vector(-2,1){78}} 
		\put(-140,95){ 0.5\%} \put(-141,97){ \vector(-2,1){76}} 
		\put(-140,85){ 0.3\%} \put(-141,87){ \vector(-2,1){75.5}} 
		\put(-140,75){ 0.2\%} \put(-141,77){ \vector(-2,1){74}} 
		\put(-210,174){\framebox(16,12){ }}
		\\
		{ (c) W28 } & { (d) 3C288 }    
	\end{tabular}
	\begin{tabular}{cccc}
		\includegraphics[trim={{.09\linewidth} {.02\linewidth} {.18\linewidth} {.05\linewidth}}, clip, width=0.24\linewidth, height = 0.21\linewidth]
		{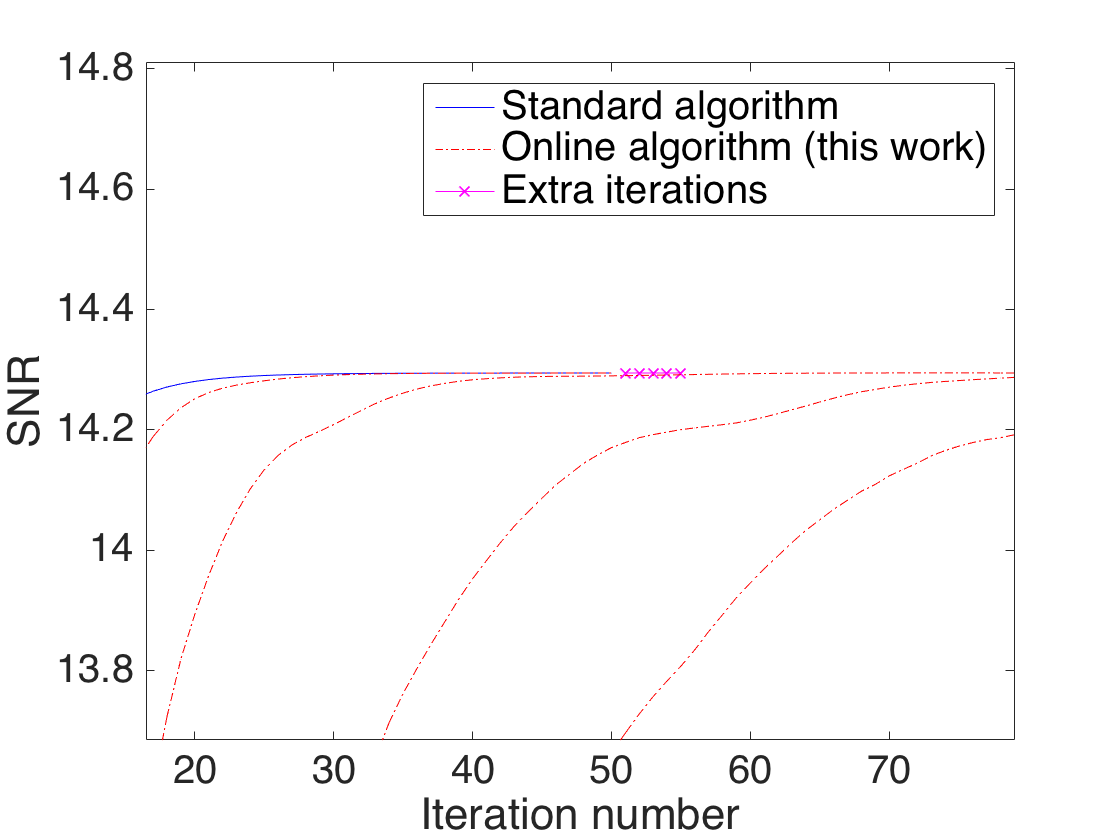} 		   
		\put(-127,48){\rotatebox{90}{SNR}} &
		\includegraphics[trim={{.09\linewidth} {.02\linewidth} {.18\linewidth} {.05\linewidth}}, clip, width=0.24\linewidth, height = 0.21\linewidth]
		{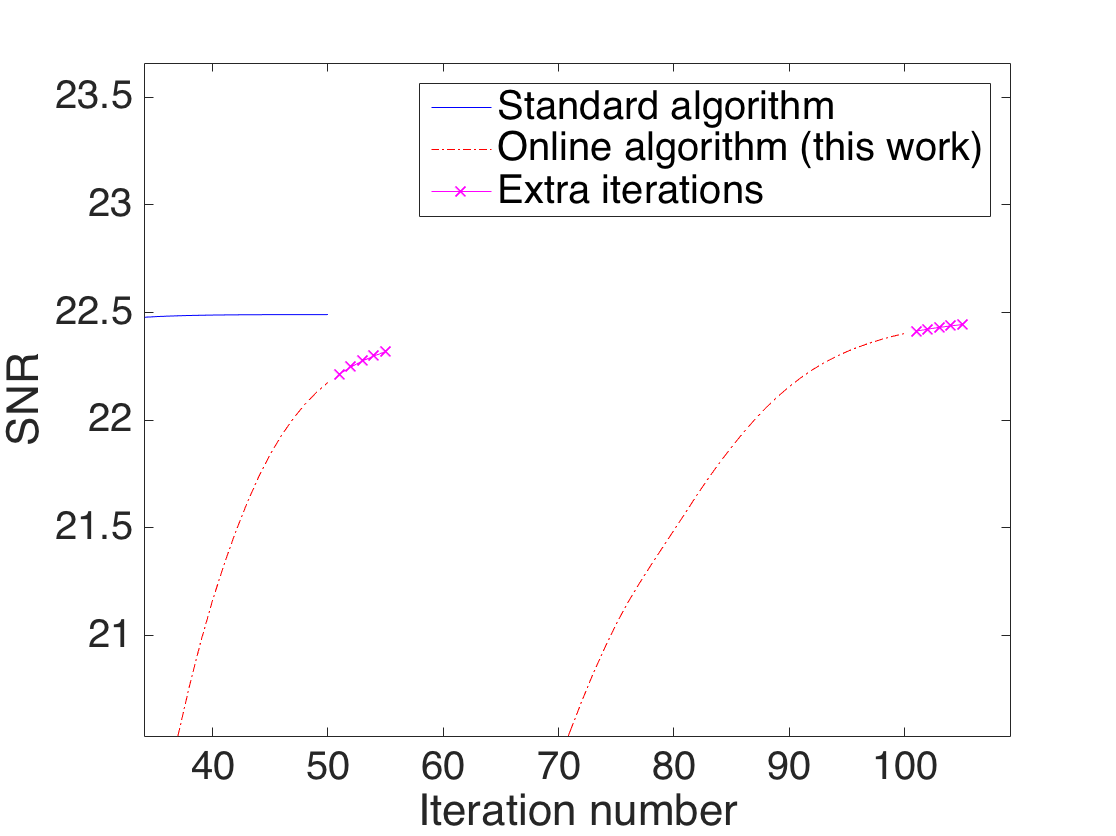}  &
		\includegraphics[trim={{.09\linewidth} {.02\linewidth} {.18\linewidth} {.05\linewidth}}, clip, width=0.24\linewidth, height = 0.21\linewidth]
		{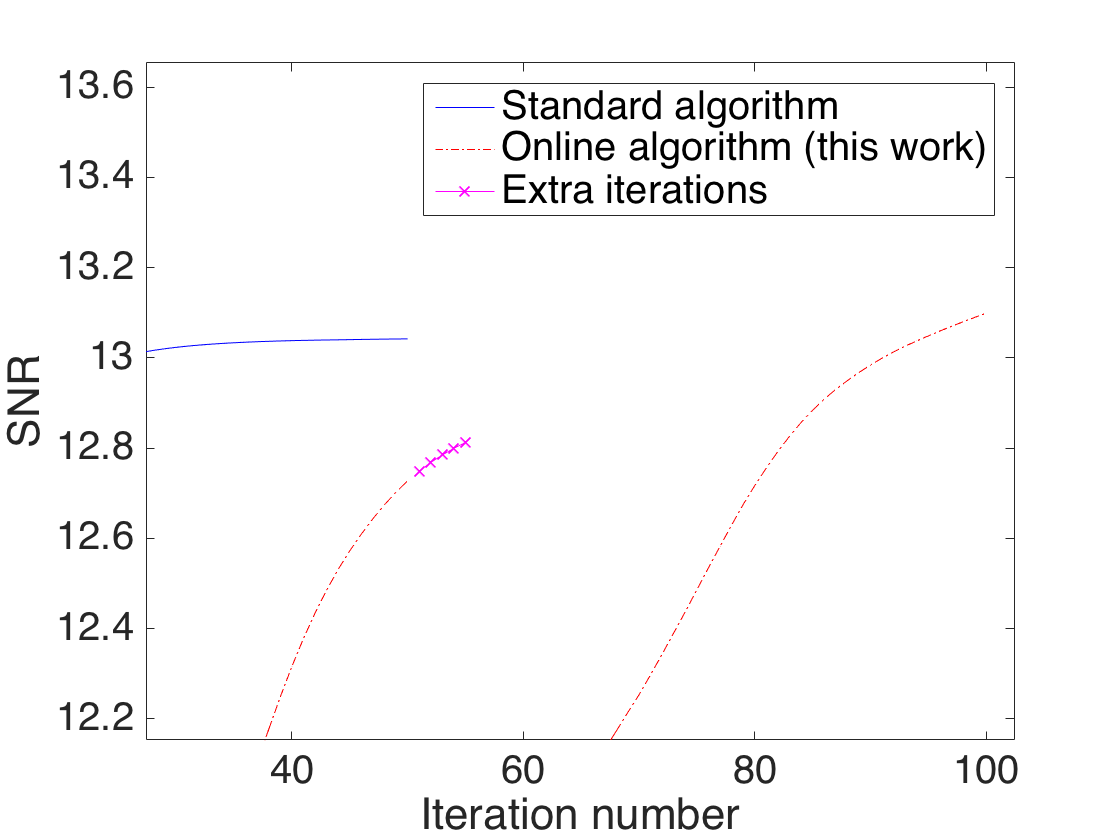}   &
		\includegraphics[trim={{.09\linewidth} {.02\linewidth} {.18\linewidth} {.05\linewidth}}, clip, width=0.24\linewidth, height = 0.21\linewidth]
		{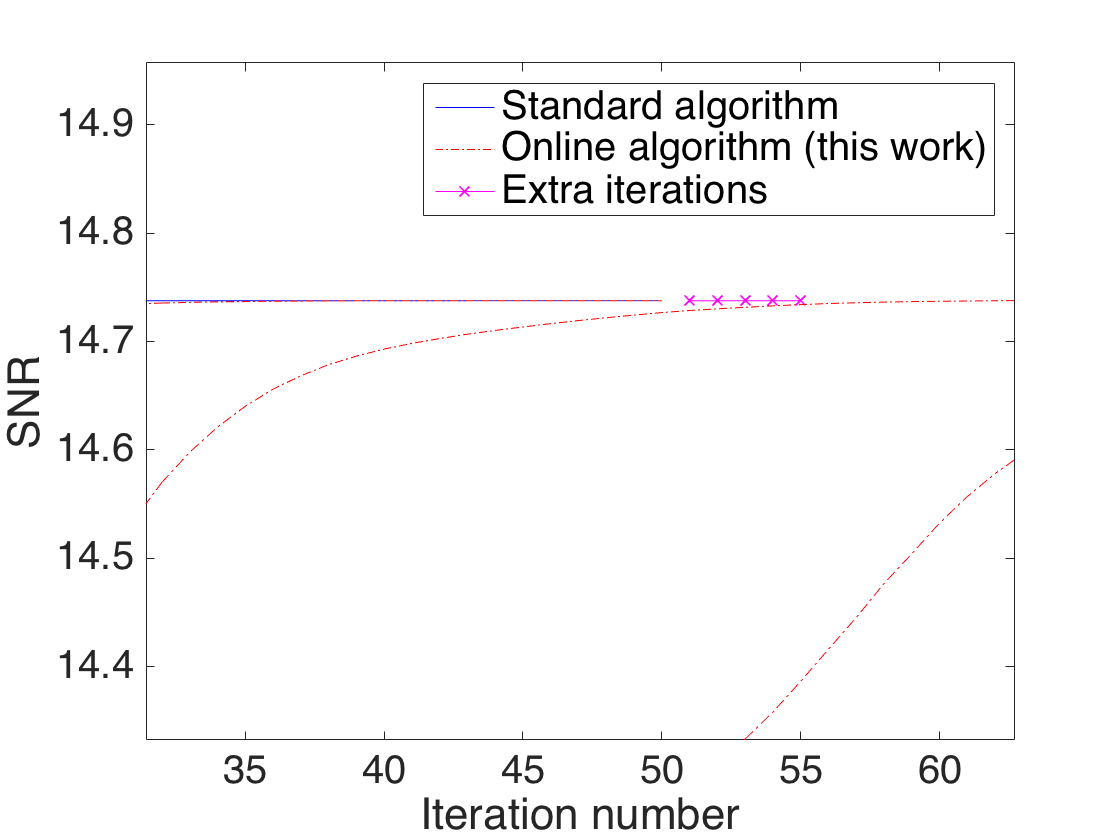}  \\
		(e) Zoom in (a) & (f) Zoom in (b) & (g) Zoom in (c) & (h) Zoom in (d) 
	\end{tabular} 
	\caption{Image reconstruction results in SNR against iteration number. The blue line and the red dot-dashed line represent the results
		of the standard algorithm and the online algorithm (our work), respectively. The magenta line with cross marks represent the 
		5 extra iterations of the online algorithm. In particular, for the online algorithm, 50, 100, 200, 300 and 500 visibility blocks are tested. 
		When the SNR of the online algorithm is less than that of the standard algorithm after the final visibility block is assimilated,
		5 extra iterations are executed (see the magenta line with cross marks).
		Panels (a)--(d): results for images of M31, Cygnus A, W28 and 3C288, respectively.
		Panels (e)--(h): zoomed in areas of the rectangles in (a)--(d), respectively.
		These plots show that both the standard and online algorithms provide reconstructed images with very similar 
		SNR. Moreover, the results of the online algorithm with respect to differing numbers of visibility blocks
		reveal that the online algorithm converges stably and is robust with respect to arbitrary numbers of visibility blocks. 
		We emphasise again that, for the online algorithm, the larger the number of blocks, the lower the visibility storage requirements. 
		Even though a large number of blocks requires lots of iterations, the first iterations are very fast due to the small amount of data used. 
		Also, since almost all the computation is done before the data acquisition finishes, the online method always ends its reconstruction task 
		much faster than the standard method. In this sense, the computation time of the online method is independent of 
		the number of blocks.
		Finally, the results of the extra iterations for the online algorithm show that an improvement can indeed be achieved but is 
		not dramatic and therefore optional; the standard iterations of the online algorithm, basically, can ensure excellent reconstructions already.
	}
	\label{fig-ite-comp-dist}
\end{figure*}
\addtolength{\tabcolsep}{\tabL}

\subsubsection{SNR analysis}

\begin{table*}
	\begin{center}
		\caption{Relative difference in SNR between the results of the standard method and the online method (our work) with the number of visibility blocks
			set to 50, 100, 200, 300 and 500, for test images of M31, Cygnus A, W28 and 3C288, for the analysis model \eqref{eqn:ir-un-af}. 
			In particular, a negative SNR means the online method performs better than the standard method (and vice versa for a positive SNR).
			The number in brackets, {\it e.g.} 50 (5), means the result of the online method is computed with 50 visibility blocks and 5 extra iterations 
			assigned after processing the entire 50 blocks. Extra iterations can improve reconstruction quality but differences are small.
			From this table, we see that, quantitatively, both methods perform similarly: sometimes the standard method is slightly better
			and sometimes the online method is, but there is no substantial difference. 	} \label{tab:snr}
		\vspace{-0.05in}
		\begin{tabular}{crrrrrrrrrr}
			\hline  
			\multirow{2}{*}{Images}  &   \multicolumn{10}{c}{Number of visibility blocks (extra number of iterations of the online method)} \\ \cline{2-11}
			  & \multicolumn{1}{c}{50} & 50 (5) & 100 & 100 (5) & 200 & 200 (5) & 300 & 300 (5) & 500 & 500 (5) 
			\\ \hline \hline
			{M31} & $1.9 {\rm e}{-6}$  & $-1.1 {\rm e}{-7}$ &  $-1.3 {\rm e}{-6}$   & $-1.3 {\rm e}{-6}$ & $-1.4 {\rm e}{-6}$
			& $-1.4 {\rm e}{-6}$  & $-1.4 {\rm e}{-6}$ &  $-1.4 {\rm e}{-6}$   &  $-1.3 {\rm e}{-6}$ &  $-1.3 {\rm e}{-6} $
			\\ \hline
			{Cygnus A  } &  $1.4 {\rm e}{-2}$  & $7.6 {\rm e}{-3}$  & $3.9 {\rm e}{-3}$   & $2.1 {\rm e}{-3}$ & $-3.0 {\rm e}{-3}$  
			&  $-3.4 {\rm e}{-3}$  & $8.7 {\rm e}{-4}$  & $5.3 {\rm e}{-4}$   & $3.1 {\rm e}{-3}$ & $2.3 {\rm e}{-3}$  
			\\ \hline
			{W28 } &  $2.4 {\rm e}{-2}$    & $1.8 {\rm e}{-2}$    & $-4.4 {\rm e}{-3}$   & $-7.0 {\rm e}{-3}$ & $1.1 {\rm e}{-2}$ 
			&  $9.6 {\rm e}{-3}$    & $1.4 {\rm e}{-2}$    & $1.2 {\rm e}{-2}$   & $1.3 {\rm e}{-2}$ & $9.3 {\rm e}{-3}$  
			\\ \hline
			{3C288 } & $1.5 {\rm e}{-6}$  & $6.1 {\rm e}{-7}$ &  $-1.3 {\rm e}{-7}$   & $-4.0 {\rm e}{-8}$  & $-6.0 {\rm e}{-8}$  
			& $-3.0 {\rm e}{-8}$  & $-7.0 {\rm e}{-8}$ &  $-3.0 {\rm e}{-8}$   & $-2.0 {\rm e}{-8}$  & $-2.0 {\rm e}{-8}$  
			\\ \hline
		\end{tabular}
	\end{center}
\end{table*}

For the online method, as we discussed, splitting the entire visibilities into different numbers of visibility blocks 
impacts both storage requirements and computational cost.
Firstly, when the number of blocks is relatively small, {\it i.e.}, similar to the number of iterations of the standard (offline) method,
the computational cost is reduced compared to the standard method.
Secondly, the larger the number of blocks, the lower the visibility storage requirements. 
Even though a large number of blocks requires lots of iterations, the
first iterations are very fast due to the small amount of data used. Moreover, no matter how large the number of blocks is,
since almost all the computation is done before the data acquisition finishes, the online method always ends its reconstruction task 
much faster than the standard method. Recall Figure \ref{fig-ite-comp-memory} for a pictorial inspection.

In order to show the influence of the number of visibility blocks on the reconstruction quality, 
simulation results are reported in Figure \ref{fig-ite-comp-dist} with, as examples, the number of visibility blocks set to
50, 100, 200, 300 and 500. The iteration histories (SNR against iterations) of the standard method and the online method are 
respectively shown by blue (solid) and red (dot-dashed) lines. 
Moreover, for the online method, five more iterations, which are regarded as extra iterations, 
are executed when the obtained SNR is lower than that of the standard method and 
are shown by the magenta line with cross marks.
For all the test images, Figure \ref{fig-ite-comp-dist} demonstrates that the online method, for different numbers of visibility blocks, 
provides reconstructions with as good SNR as those achieved by the standard method, {\it i.e.}, slightly higher or lower
with almost equal chances but with no substantial difference.  

In particular, from Figure \ref{fig-ite-comp-dist}, we can also see that the larger the visibility block size, 
the quicker (in terms of number of iterations) the highest SNR is reached. Nevertheless, after processing the last visibility block, 
all of the settings with different number of blocks get reconstructions with very similar SNR. 
This suggests that the online method converges stably and is robust with respect 
to arbitrary numbers of visibility blocks. 

Accompanying Figure \ref{fig-ite-comp-dist}, Table \ref{tab:snr} gives the relative difference in SNR between the results of 
the standard method (${\rm SNR}_{\rm standard}$) and online method (${\rm SNR}_{\rm online}$) under 
different number of visibility blocks, {\it i.e.},
\begin{equation} \label{eqn:rel-diff}
	{\rm relative \ difference} = \frac{{\rm SNR}_{\rm standard}  - {\rm SNR}_{\rm online}}{{\rm SNR}_{\rm standard}}.
\end{equation}
Positive and negative signs of the relative difference correspond to 
better performance of the standard method and the online method, respectively.
In the previous subsection, we concluded that both methods achieve comparable reconstruction results via visual validation.
This agrees with Table \ref{tab:snr}, which shows that both methods quantitatively perform similarly:
sometimes the standard method is slightly better and sometimes the online method is, but there are no substantial differences in quality. 

Figure \ref{fig-ite-comp-dist} and Table \ref{tab:snr} tells us that the online method can already provide very good reconstructions 
after processing the last visibility block. After that, it is optional to execute a few more iterations to improve the SNR of the reconstruction,
as we can see from the magenta lines with cross marks in Figure \ref{fig-ite-comp-dist}. 
However, the improvement is not dramatic; the standard number of iterations, basically, can ensure excellent reconstructions already.

\section{Conclusions}\label{sec:con}
The work in this article has been motivated by critical computational problems in scaling RI imaging to the big-data
era of radio astronomy that will be ushered in by the SKA and precursor telescopes.  
In particular, we addressed the extremely high storage requirements and computational costs of standard (offline) methods of
recovering images from the raw data that will be acquired by forthcoming telescopes. 
We presented an online imaging methodology by extending standard sparse regularisation methods. 

Generally speaking, our online method starts the reconstruction task at the beginning of the data acquisition stage (not after) 
and keeps updating the quality of the reconstruction by continually assimilating newly acquired visibilities (visibility blocks), 
before discarding them to release storage. In other words, it combines the data acquisition stage with the data processing stage, 
and it orderly processes data blocks as received at consecutive time slots. In detail, the online method firstly achieves
intermediate reconstructions using the currently acquired data blocks, and then treats the intermediate reconstruction
as a starting point to further update the reconstruction with newly obtained data blocks, until the last data block is processed. 

The online method achieves good reconstruction fidelity much faster than standard methods, 
which do not begin until the visibility acquisition stage is completed. 
Roughly speaking, the online method has the ability of providing 
an excellent reconstruction as soon as the visibility acquisition procedure completes, which significantly improves 
the reconstruction speed.
Moreover, the computational cost of the online method is further reduced for a reasonable choice of number of blocks
since the amount of data to be considered for early iterations is small.
Furthermore, the online method has the advantage of significantly lower visibility storage requirements. 
In principle, the storage needed for the online method can be arbitrarily small; recall that standard methods
always require all the visibilities to be stored for subsequent processing.
Consequently, these two main virtues -- extremely low storage requirements and fast computation speed -- 
make the online method highly relevant for addressing the big-data processing obstacles of RI imaging in the near future. 

There are a number of avenues of future work. Since the proposed online framework is very general, 
it will be interesting to investigate equipping other methods with this online strategy.
Considering overcomplete bases in the objective functionals is likely to provide improvements in reconstruction fidelity and
is another interesting avenue of future investigation. 
The online method will be implemented in the existing PURIFY\footnote{\url{https://github.com/basp-group/purify}} 
package for RI imaging, where it may then be applied easily to real observations and combined with existing performance 
gains from distributed and shared parallelisation. 
Finally, we will integrate our online method with the uncertainty quantification framework presented in \cite{CPM17B} to 
perform efficient imaging and uncertainty quantification for the emerging big-data era of radio astronomy.

\section*{Acknowledgements}
This work is supported by the UK Engineering and Physical Sciences Research Council (EPSRC) by grant EP/M011089/1 
and the Science and Technology Facilities Council (STFC) by grant ST/N000811/1.

\bibliographystyle{mnras}
\bibliography{refs_xhcai}

\begin{thebibliography}{}
\makeatletter
\relax
\def\mn@urlcharsother{\let\do\@makeother \do\$\do\&\do\#\do\^\do\_\do\%\do\~}
\def\mn@doi{\begingroup\mn@urlcharsother \@ifnextchar [ {\mn@doi@}
  {\mn@doi@[]}}
\def\mn@doi@[#1]#2{\def\@tempa{#1}\ifx\@tempa\@empty \href
  {http://dx.doi.org/#2} {doi:#2}\else \href {http://dx.doi.org/#2} {#1}\fi
  \endgroup}
\def\mn@eprint#1#2{\mn@eprint@#1:#2::\@nil}
\def\mn@eprint@arXiv#1{\href {http://arxiv.org/abs/#1} {{\tt arXiv:#1}}}
\def\mn@eprint@dblp#1{\href {http://dblp.uni-trier.de/rec/bibtex/#1.xml}
  {dblp:#1}}
\def\mn@eprint@#1:#2:#3:#4\@nil{\def\@tempa {#1}\def\@tempb {#2}\def\@tempc
  {#3}\ifx \@tempc \@empty \let \@tempc \@tempb \let \@tempb \@tempa \fi \ifx
  \@tempb \@empty \def\@tempb {arXiv}\fi \@ifundefined
  {mn@eprint@\@tempb}{\@tempb:\@tempc}{\expandafter \expandafter \csname
  mn@eprint@\@tempb\endcsname \expandafter{\@tempc}}}

\bibitem[\protect\citeauthoryear{{Ables}}{{Ables}}{1974}]{A74}
{Ables} J.~G.,  1974, \aaps, 15, 383

\bibitem[\protect\citeauthoryear{Bauschke \& Combettes}{Bauschke \&
  Combettes}{2011}]{BC11}
Bauschke H.~H.,  Combettes P.~L.,  2011, Convex Analysis and Monotone Operator
  Theory in Hilbert Spaces.
Springer-Verlag

\bibitem[\protect\citeauthoryear{{Bhatnagar} \& {Corwnell}}{{Bhatnagar} \&
  {Corwnell}}{2004}]{BC04}
{Bhatnagar} S.,  {Corwnell} T.~J.,  2004, A\&A, 426, 747

\bibitem[\protect\citeauthoryear{Broekema, {van Nieuwpoort}  \& Bal}{Broekema
  et~al.}{2015}]{BNB15}
Broekema P.~C.,  {van Nieuwpoort} R.~V.,   Bal H.~E.,  2015, J. Instrum., 10,
  C07004

\bibitem[\protect\citeauthoryear{{Cai}, {Pereyra}  \& {McEwen}}{{Cai}
  et~al.}{2017b}]{CPM17}
{Cai} X.,  {Pereyra} M.,   {McEwen} J.~D.,  2017b, preprint (\mn@eprint {arXiv}
  {1711.04818})

\bibitem[\protect\citeauthoryear{{Cai}, {Pereyra}  \& {McEwen}}{{Cai}
  et~al.}{2017a}]{CPM17B}
{Cai} X.,  {Pereyra} M.,   {McEwen} J.~D.,  2017a, preprint (\mn@eprint {arXiv}
  {1711.04819})

\bibitem[\protect\citeauthoryear{{Carrillo}, {McEwen}  \& {Wiaux}}{{Carrillo}
  et~al.}{2012}]{CMW12}
{Carrillo} R.~E.,  {McEwen} J.~D.,   {Wiaux} Y.,  2012, \mn@doi [\mnras]
  {10.1111/j.1365-2966.2012.21605.x}, 426, 1223

\bibitem[\protect\citeauthoryear{{Carrillo}, {McEwen}  \& {Wiaux}}{{Carrillo}
  et~al.}{2014}]{car14}
{Carrillo} R.~E.,  {McEwen} J.~D.,   {Wiaux} Y.,  2014, \mn@doi [\mnras]
  {10.1093/mnras/stu202}, \href
  {http://adsabs.harvard.edu/abs/2014MNRAS.439.3591C} {439, 3591}

\bibitem[\protect\citeauthoryear{Cleju, Jafari  \& Plumbley}{Cleju
  et~al.}{2012}]{CJP12}
Cleju N.,  Jafari M.~G.,   Plumbley M.~D.,  2012, in Signal Processing
  Conference (EUSIPCO). pp 869--873

\bibitem[\protect\citeauthoryear{Combettes \& Pesquet}{Combettes \&
  Pesquet}{2010}]{CP10}
Combettes P.~L.,  Pesquet J.~C.,  2010, arXiv:0912.3522v4

\bibitem[\protect\citeauthoryear{Combettes \& Wajs}{Combettes \&
  Wajs}{2005}]{CW05}
Combettes P.~L.,  Wajs V.~R.,  2005, Multiscale Model. Simul., 4, 1168

\bibitem[\protect\citeauthoryear{{Cornwell}}{{Cornwell}}{2008}]{cor08}
{Cornwell} T.~J.,  2008, \mn@doi [IEEE J. Sel. Topics Signal Process.]
  {10.1109/JSTSP.2008.2006388}, \href
  {http://adsabs.harvard.edu/abs/2008ISTSP...2..793C} {2, 793}

\bibitem[\protect\citeauthoryear{{Cornwell} \& {Evans}}{{Cornwell} \&
  {Evans}}{1985}]{CE85}
{Cornwell} T.~J.,  {Evans} K.~F.,  1985, A\&A, 143, 77

\bibitem[\protect\citeauthoryear{{Dabbech}, {Ferrari}, {Mary}, {Slezak},
  {Smirnov}  \& {Kenyon}}{{Dabbech} et~al.}{2015}]{dab15}
{Dabbech} A.,  {Ferrari} C.,  {Mary} D.,  {Slezak} E.,  {Smirnov} O.,
  {Kenyon} J.~S.,  2015, \mn@doi [\aap] {10.1051/0004-6361/201424602}, \href
  {http://adsabs.harvard.edu/abs/2015A%26A...576A...7D} {576, A7}

\bibitem[\protect\citeauthoryear{Dabbech, Onose, Abdulaziz, Perley, M.~Smirnov
  \& Wiaux}{Dabbech et~al.}{2017b}]{DOAPMW17}
Dabbech A.,  Onose A.,  Abdulaziz A.,  Perley R.,  M.~Smirnov O.,   Wiaux Y.,
  2017b, arXiv:1710.08810

\bibitem[\protect\citeauthoryear{Dabbech, Wolz, Pratley, McEwen  \&
  Wiaux}{Dabbech et~al.}{2017a}]{DWPMW17}
Dabbech A.,  Wolz L.,  Pratley L.,  McEwen J.~D.,   Wiaux Y.,  2017a,
  arXiv:1702.05009

\bibitem[\protect\citeauthoryear{Dewdney, Turner, Millenaar, McCool, Lazio  \&
  Cornwell}{Dewdney et~al.}{2013}]{dew13}
Dewdney P.,  Turner W.,  Millenaar R.,  McCool R.,  Lazio J.,   Cornwell T.,
  2013, Document number SKA-TEL-SKO-DD-001 Revision, 1

\bibitem[\protect\citeauthoryear{Elad, Milanfar  \& Rubinstein}{Elad
  et~al.}{2007}]{EMR07}
Elad M.,  Milanfar P.,   Rubinstein R.,  2007, Inv. Prob., 23, 947

\bibitem[\protect\citeauthoryear{Fadili \& Starck}{Fadili \&
  Starck}{2009}]{FS09}
Fadili M.~J.,  Starck J.~L.,  2009, in ICIP.

\bibitem[\protect\citeauthoryear{{Fessler} \& {Sutton}}{{Fessler} \&
  {Sutton}}{2003}]{fes03}
{Fessler} J.~A.,  {Sutton} B.~P.,  2003, \mn@doi [IEEE Trans. Signal Process.]
  {10.1109/TSP.2002.807005}, \href
  {http://adsabs.harvard.edu/abs/2003ITSP...51..560F} {51, 560}

\bibitem[\protect\citeauthoryear{{Garsden} et~al.,}{{Garsden}
  et~al.}{2015}]{gar15}
{Garsden} H.,  et~al., 2015, \mn@doi [\aap] {10.1051/0004-6361/201424504},
  \href {http://adsabs.harvard.edu/abs/2015A%26A...575A..90G} {575, A90}

\bibitem[\protect\citeauthoryear{{Gull} \& {Daniell}}{{Gull} \&
  {Daniell}}{1978}]{GD78}
{Gull} S.~F.,  {Daniell} G.~J.,  1978, Nature, 272, 686

\bibitem[\protect\citeauthoryear{{Hazan}}{{Hazan}}{2015}]{H15}
{Hazan} E.,  2015, Foundations and Trends in Optimization, 2, 157

\bibitem[\protect\citeauthoryear{{H{\"o}gbom}}{{H{\"o}gbom}}{1974}]{hog74}
{H{\"o}gbom} J.~A.,  1974, \aaps, \href
  {http://adsabs.harvard.edu/abs/1974A%26AS...15..417H} {15, 417}

\bibitem[\protect\citeauthoryear{{Hotan} et~al.,}{{Hotan} et~al.}{2014}]{hot14}
{Hotan} A.~W.,  et~al., 2014, \mn@doi [\pasa] {10.1017/pasa.2014.36}, \href
  {http://adsabs.harvard.edu/abs/2014PASA...31...41H} {31, e041}

\bibitem[\protect\citeauthoryear{Jacques, Hammond  \& Fadili}{Jacques
  et~al.}{2011}]{JHF11}
Jacques L.,  Hammond D.,   Fadili M.,  2011, IEEE Trans. Inf. Theory, 57, 559

\bibitem[\protect\citeauthoryear{{Johnston-Hollitt} et~al.,}{{Johnston-Hollitt}
  et~al.}{2015}]{joh15}
{Johnston-Hollitt} M.,  et~al., 2015, Advancing Astrophysics with the Square
  Kilometre Array (AASKA14), \href
  {http://adsabs.harvard.edu/abs/2015aska.confE..92J} {p.~92}

\bibitem[\protect\citeauthoryear{Kartik, Dabbech, Thiran  \& Y.}{Kartik
  et~al.}{2017a}]{KDTW17}
Kartik S.~V.,  Dabbech A.,  Thiran J.-P.,   Y. W.,  2017a, arXiv:1709.03950

\bibitem[\protect\citeauthoryear{Kartik, Carrillo, Thiran  \& Y.}{Kartik
  et~al.}{2017b}]{KCTW17}
Kartik S.~V.,  Carrillo R.~E.,  Thiran J.-P.,   Y. W.,  2017b, Springer,
  New-York, 468, 2382

\bibitem[\protect\citeauthoryear{{Koopmans} et~al.,}{{Koopmans}
  et~al.}{2015}]{koo15}
{Koopmans} L.,  et~al., 2015, Advancing Astrophysics with the Square Kilometre
  Array (AASKA14), \href {http://adsabs.harvard.edu/abs/2015aska.confE...1K}
  {p.~1}

\bibitem[\protect\citeauthoryear{{Li}, {Cornwell}  \& {de Hoog}}{{Li}
  et~al.}{2011a}]{li11a}
{Li} F.,  {Cornwell} T.~J.,   {de Hoog} F.,  2011a, \mn@doi [A\&A]
  {10.1051/0004-6361/201015045}, \href
  {http://adsabs.harvard.edu/abs/2011A%26A...528A..31L} {528, A31}

\bibitem[\protect\citeauthoryear{{Li}, {Brown}, {Cornwell}  \& {de Hoog}}{{Li}
  et~al.}{2011b}]{li11b}
{Li} F.,  {Brown} S.,  {Cornwell} T.~J.,   {de Hoog} F.,  2011b, \mn@doi [\aap]
  {10.1051/0004-6361/201015890}, \href
  {http://adsabs.harvard.edu/abs/2011A%26A...531A.126L} {531, A126}

\bibitem[\protect\citeauthoryear{Maisinger, Hobson  \& Lasenby}{Maisinger
  et~al.}{2004}]{MHL04}
Maisinger K.,  Hobson M.~P.,   Lasenby A.~N.,  2004, \mnras, 347, 339

\bibitem[\protect\citeauthoryear{{McEwen} \& {Wiaux}}{{McEwen} \&
  {Wiaux}}{2011}]{mce11}
{McEwen} J.~D.,  {Wiaux} Y.,  2011, \mn@doi [\mnras]
  {10.1111/j.1365-2966.2011.18217.x}, \href
  {http://adsabs.harvard.edu/abs/2011MNRAS.413.1318M} {413, 1318}

\bibitem[\protect\citeauthoryear{Moreau}{Moreau}{1965}]{M65}
Moreau J.~J.,  1965, Bulletin de la Soci\'et\'e Math\'ematique de France, 93,
  273

\bibitem[\protect\citeauthoryear{{Offringa} et~al.,}{{Offringa}
  et~al.}{2014}]{off14}
{Offringa} A.~R.,  et~al., 2014, \mn@doi [\mnras] {10.1093/mnras/stu1368},
  \href {http://adsabs.harvard.edu/abs/2014MNRAS.444..606O} {444, 606}

\bibitem[\protect\citeauthoryear{Onose, Carrillo, Repetti, McEwen, Thiran,
  Pesquet  \& Wiaux}{Onose et~al.}{2016}]{OCRMTPW16}
Onose A.,  Carrillo R.~E.,  Repetti A.,  McEwen J.~D.,  Thiran J.~P.,  Pesquet
  J.~C.,   Wiaux Y.,  2016, \mnras, 462, 4314

\bibitem[\protect\citeauthoryear{Onose, Dabbech  \& Wiaux}{Onose
  et~al.}{2017}]{ODW17}
Onose A.,  Dabbech A.,   Wiaux Y.,  2017, \mnras, 469, 938

\bibitem[\protect\citeauthoryear{{Pratley} \& {Johnston-Hollitt}}{{Pratley} \&
  {Johnston-Hollitt}}{2016}]{pra16}
{Pratley} L.,  {Johnston-Hollitt} M.,  2016, \mn@doi [\mnras]
  {10.1093/mnras/stw1377}, \href
  {http://adsabs.harvard.edu/abs/2016MNRAS.tmp.1019P} {}

\bibitem[\protect\citeauthoryear{{Pratley}, {McEwen}, {d'Avezac}, {Carrillo},
  {Onose}  \& {Wiaux}}{{Pratley} et~al.}{2018}]{PMdCOW16}
{Pratley} L.,  {McEwen} J.~D.,  {d'Avezac} M.,  {Carrillo} R.~E.,  {Onose} A.,
   {Wiaux} Y.,  2018, \mn@doi [\mnras] {10.1093/mnras/stx2237}, \href
  {http://adsabs.harvard.edu/abs/2018MNRAS.473.1038P} {473, 1038}

\bibitem[\protect\citeauthoryear{Puy, Vandergheynst  \& Wiaux}{Puy
  et~al.}{2011}]{PVW11}
Puy G.,  Vandergheynst P.,   Wiaux Y.,  2011, IEEE Signal Processing Letters,
  18, 595

\bibitem[\protect\citeauthoryear{{Rau}, Bhatnagar, Voronkov  \& Cornwell}{{Rau}
  et~al.}{2009}]{RBVC09}
{Rau} U.,  Bhatnagar S.,  Voronkov M.~A.,   Cornwell T.~J.,  2009, in Proc.
  IEEE. pp 1472--1481

\bibitem[\protect\citeauthoryear{{Ryle} \& {Hewish}}{{Ryle} \&
  {Hewish}}{1960}]{ryl60}
{Ryle} M.,  {Hewish} A.,  1960, \mn@doi [\mnras] {10.1093/mnras/120.3.220},
  \href {http://adsabs.harvard.edu/abs/1960MNRAS.120..220R} {120, 220}

\bibitem[\protect\citeauthoryear{{Ryle} \& {Vonberg}}{{Ryle} \&
  {Vonberg}}{1946}]{RV46}
{Ryle} M.,  {Vonberg} D.~D.,  1946, \mn@doi [Nature] {10.1038/158339b0}, 158,
  339

\bibitem[\protect\citeauthoryear{{Shalev-Shwartz}}{{Shalev-Shwartz}}{2011}]{S11}
{Shalev-Shwartz} S.,  2011, Foundations and Trends in Machine Learning, 2, 107

\bibitem[\protect\citeauthoryear{{Stewart}, {Fenech}  \& {Muxlow}}{{Stewart}
  et~al.}{2011}]{SFM11}
{Stewart} I.~M.,  {Fenech} D.~M.,   {Muxlow} T.~W.~B.,  2011, A\&A, 535, A81

\bibitem[\protect\citeauthoryear{Thompson, Moran  \& Swenson}{Thompson
  et~al.}{2008}]{tho08}
Thompson A.,  Moran J.,   Swenson G.,  2008, Interferometry and Synthesis in
  Radio Astronomy.
Wiley, \url {https://books.google.co.uk/books?id=S99Z7kNBYzAC}

\bibitem[\protect\citeauthoryear{{Tingay} et~al.,}{{Tingay}
  et~al.}{2013}]{tin13}
{Tingay} S.~J.,  et~al., 2013, \mn@doi [\pasa] {10.1017/pasa.2012.007}, \href
  {http://adsabs.harvard.edu/abs/2013PASA...30....7T} {30, 7}

\bibitem[\protect\citeauthoryear{{Wiaux}, {Jacques}, {Puy}, {Scaife}  \&
  {Vandergheynst}}{{Wiaux} et~al.}{2009a}]{wia09a}
{Wiaux} Y.,  {Jacques} L.,  {Puy} G.,  {Scaife} A.~M.~M.,   {Vandergheynst} P.,
   2009a, \mn@doi [\mnras] {10.1111/j.1365-2966.2009.14665.x}, \href
  {http://adsabs.harvard.edu/abs/2009MNRAS.395.1733W} {395, 1733}

\bibitem[\protect\citeauthoryear{{Wiaux}, {Puy}, {Boursier}  \&
  {Vandergheynst}}{{Wiaux} et~al.}{2009b}]{wia09b}
{Wiaux} Y.,  {Puy} G.,  {Boursier} Y.,   {Vandergheynst} P.,  2009b, \mn@doi
  [\mnras] {10.1111/j.1365-2966.2009.15519.x}, \href
  {http://adsabs.harvard.edu/abs/2009MNRAS.400.1029W} {400, 1029}

\bibitem[\protect\citeauthoryear{{Wolz}, {McEwen}, {Abdalla}, {Carrillo}  \&
  {Wiaux}}{{Wolz} et~al.}{2013}]{wol13}
{Wolz} L.,  {McEwen} J.~D.,  {Abdalla} F.~B.,  {Carrillo} R.~E.,   {Wiaux} Y.,
  2013, \mn@doi [\mnras] {10.1093/mnras/stt1707}, \href
  {http://adsabs.harvard.edu/abs/2013MNRAS.436.1993W} {436, 1993}

\bibitem[\protect\citeauthoryear{{van Haarlem} et~al.,}{{van Haarlem}
  et~al.}{2013}]{van13}
{van Haarlem} M.~P.,  et~al., 2013, \mn@doi [\aap]
  {10.1051/0004-6361/201220873}, \href
  {http://adsabs.harvard.edu/abs/2013A%26A...556A...2V} {556, A2}

\makeatother
\end{thebibliography}

\label{lastpage}

\end{document}